\documentclass[acmsmall,screen,nonacm]{acmart}

\usepackage[utf8]{inputenc}
\usepackage{amsmath}
\usepackage[shortlabels]{enumitem} 

\usepackage[normalem]{ulem}
\usepackage{mathpartir}
\usepackage{proof}
\usepackage{mathtools}
\usepackage{longtable}
\usepackage{supertabular}
\usepackage{ stmaryrd }

\usepackage{geometry}
\usepackage{changepage}
\usepackage{longtable}
\usepackage{comment}
\usepackage{color}   
\usepackage{hyperref}
\hypersetup{
    colorlinks,
    citecolor=black,
    filecolor=black,
    linkcolor={blue!50!black},
    urlcolor=black
}

\usepackage{tikz} 

\usepackage{diagbox}
\usepackage{pdflscape}

\usepackage{makecell}
\usepackage{changepage}
\usepackage{xcolor}
\usepackage{multicol}
\usepackage{multirow}

\usepackage{algorithm}

\makeindex

\newtheorem{theorem}{Theorem}[section]
\newtheorem{lemma}[theorem]{Lemma}

\newtheorem{corollary}[theorem]{Corollary}

\newcommand\cotype{\operatorname{\mathsf{cotype}}}
\newcommand\type{\operatorname{\mathsf{type}}}

\newcommand{\m}[1]{\operatorname{\mathsf{{#1}}}}
\newcommand{\mi}[1]{\mathop{\operatorname{\mathit{{#1}}}}}

\newcommand{\mt}[1]{\mathop{\operatorname{\mathtt{{#1}}}}}

\newcommand\hra\hookrightarrow

\newcommand\bdepth[1]{_{({#1})}}

\newcommand{\depth}[1]{\bdepth{#1}}

\newcounter{numbered}
\newenvironment{numbered}{
    \setcounter{numbered}{0}
    }{
    \setcounter{numbered}{0}
    }

\newcommand{\previousNumber}{\the\numexpr\value{numbered}-1\relax }

\newcommand{\blue}[1]{\textcolor{blue}{#1}}

\newcommand{\symbrec}{\medcirc}
\newcommand{\hrec}{^\symbrec}
\newcommand{\symbcontra}{\blacksquare}
\newcommand{\hcontra}{^\symbcontra}

\renewcommand{\cite}{\citep}

\newcommand{\recvar}[1]{\underline{\mi{#1}}}

\newcounter{rulenumber}
\newcommand{\resetRN}{\setcounter{rulenumber}{0}}
\newcommand{\RN}{\refstepcounter{rulenumber}\textsc{(\therulenumber)}}
\makeatletter
\newcommand{\RNwithLabel}[1]{\refstepcounter{rulenumber}%
\textsc{({#1})}
\def\@currentlabel{\textsc{{#1}}}%
}
\makeatother

\hypersetup{
  breaklinks=true,   
  colorlinks=true,   
  pdfusetitle=true,  
}
\acmConference[Conference acronym 'XX]{Make sure to enter the correct
  conference title from your rights confirmation emai}{June 03--05,
  2018}{Woodstock, NY}
\acmISBN{978-1-4503-XXXX-X/18/06}

\author{Zhibo Chen}
\orcid{0000-0003-0045-5024}
\email{zhiboc@andrew.cmu.edu}
\author{Frank Pfenning}
\orcid{0000-0002-8279-5817}
\email{fp@cs.cmu.edu}
\affiliation{%
  \institution{Carnegie Mellon University}
  \country{USA}
}

\ccsdesc[500]{Theory of computation~Proof theory}
\ccsdesc[500]{Theory of computation~Equational logic and rewriting}

\keywords{Pattern Unification, Rational Terms, Regular B\"ohm Trees}
\begin{document}

\title{A Saturation-Based Unification Algorithm for Higher-Order Rational Patterns}

\begin{abstract}
Higher-order unification has been shown to be undecidable \cite{Huet73ic}. Miller 
discovered the pattern fragment and subsequently showed that 
higher-order pattern unification is decidable and has most general unifiers \citeyearpar{Miller91jlc}.
We extend the algorithm to higher-order rational terms (a.k.a. regular B\"{o}hm trees \cite{Huet98mscs}, a form of 
cyclic $\lambda$-terms)
and show that pattern unification on higher-order rational terms is decidable and 
has most general unifiers. We prove the soundness 
and completeness of the algorithm.
\end{abstract}

\maketitle

\section{Introduction}

Unification is the backbone of logic programming \citep{Miller91jlc} and is also used in type reconstruction in 
the implementation and coverage checking of dependent type theories \citep{Pfenning99cade,Schurmann03tphol}.
Given a list of equations with unification metavariables, unification is the problem of finding substitutions for the
unification metavariables such that the equations hold true. Often, one is interested in finding most general unifiers. 
Consider the following equation as an example, where unification metavariables are written in capital letters.
\[\lambda x. \, \lambda y.\, \lambda z.\, F\, x\, y \doteq 
\lambda x. \, \lambda y.\, \lambda z.\, G\, y\, z 
\]

The most general unifier in this case is $F = \lambda v.\, \lambda w.\, H\, v$ and 
$G = \lambda u.\, \lambda v.\, H\, v$ for some fresh unification metavariable $H$.

The attempt to develop a formal representation for circular proof systems \cite{Brotherston11jlc,Fortier13csl} has 
led to the development of CoLF \cite{Chen23fossacs}, a logical framework with higher-order rational (circular) terms. Type checking and reconstruction in CoLF involves unification on
inherently circular terms. 
We also foresee applications of our unification algorithms
in the context of cyclic logic and process calculi \cite{Derakhshan19arxiv}, for example,
to implement transformations of recursive (cyclic) processes or to
perform type inference in the presence of recursive (cyclic)
types.

In this paper, we provide a unification algorithm on higher-order rational terms in 
the sense of the type theory CoLF \cite{Chen23fossacs}, 
where two $\lambda$-terms are equal if their infinite unfoldings as rational trees are equal. 
The higher-order rational terms we are considering are also called $\bot$-free regular B\"ohm trees by \citet{Huet98mscs}. 
Our work is distinguished from recent works on nominal unification in $\lambda$-calculus with recursive let \cite{Schmidt-Schauss22fi} (a.k.a. cyclic $\lambda$-calculus), in that 
the notion of equality in their work is much weaker than ours. Their equality is based on alpha-equivalence 
and permutation of order of declaration within the recursive let construct, but our equality is based on the infinite tree equalities
generated from circular terms.
For instance, given two recursive definitions $\recvar r =_d c\, \recvar r$ and $\recvar s =_d c\, (c\, \recvar s)$, 
our algorithm considers $\recvar r$ and $\recvar s$ to be equal, whereas Schmidt-Schau{\ss} et. al.'s algorithm distinguishes these two terms.

We only consider the case of unification problems between simply-typed higher-order rational terms, and in particular, we 
treat validity as a separate issue and thus do not distinguish between $\type$ and $\cotype$ \cite{Chen23fossacs}. 
For instance, when encountering the unification problem $F = \m{succ}\, F$, supposedly 
with the type of natural numbers, our algorithm is happy to come up with the solution $F$ being 
an infinite stack of $\m{succ}$'s, and disregards the fact the circular terms $F =\recvar{r}, \recvar{r} =_d \m{succ}\, \recvar{r}$ are not valid.
In an implementation, validity checking can be a separate procedure from unification.
The approach is to build validity checking into unification, such that 
 $F = \m{succ}\, F$ of natural number type has no unifier due to the failure of the occurs check, because validity 
 requires every term of natural number type to be a finite term. We did not take the combined approach as we find separating
 two issues leads to a cleaner presentation of the unification algorithm.

One of the simplest examples of higher-order unification on higher-order rational terms is the following one, where $F$ is of simple type $(* \to *) \to *$.
\[\lambda x.\, x\, (F\, x) \doteq \lambda x.\, F\, x\]
If we were to consider this problem in the setting of non-cyclic unification, there would be no unifier 
due to the failure of the occurs check. However, 
our cyclic-unification algorithm will successfully find the unifier $F = \lambda x.\, \recvar r\, x$ where $\recvar r =_d \lambda x.\, x\, (\recvar r\, x)$.
The symbol $\recvar r$ is a recursion constant that  unfolds to the infinitary 
term, $\lambda x.\, x\, (x\, (x\, (x\, \dots)))$, an abstraction that binds $x$ and its body is an infinite stack of $x$'s.

We first present the algorithm for first-order rational unification in a new way (Section \ref{sec:fo_case}) and then extend the algorithm to include higher-order patterns (Section~\ref{sec:ho_case}). 
In each case, we
 first define the unification problem (Section~\ref{sec:fo_problem_formulation} and \ref{sec:problem_formulation}) and 
 give a preprocessing algorithm that 
transform an arbitrary unification problem into a flattened form (Section~\ref{sec:fo_preprocessing} and \ref{sec:preprocessing}).
Then, we give a saturation-based 
algorithm that operates on the flattened form (Section~\ref{sec:fo_algorithm} and \ref{sec:algorithm}). The saturation rules are complemented with 
examples showing how the algorithm operates on concrete problems. Finally, a proof of the correctness of the 
algorithm is given (Section~\ref{sec:fo_metatheory} and \ref{sec:metatheory}). 
Examples use the syntax of Twelf \cite{Pfenning99cade} extended 
with cyclic terms of CoLF \cite{Chen23fossacs}.

\section{First-Order Rational Unification}
\label{sec:fo_case}

First-order rational unification \cite{Jaffar84ngc} arises directly out of first-order unification \cite{Robinson65jacm}, but without occurs check.
We give a new presentation of the algorithm based on saturation \cite{Ganzinger96icalp,Pfenning06lnUnif}, that mimics the structure 
of the higher-order case in Section~\ref{sec:ho_case}. In Jaffar's algorithm, the unification is presented as 
transformations on equivalence classes of terms (containing variables) into a solved form, 
where solutions can be read directly. Our algorithm shares the essential idea as his algorithm, but 
we presented it very differently.
The primary motivation for a different presentation is to make the later presentation of the higher-order case easier to follow.
The circular terms in Jaffar's algorithm are created implicitly by the presence of a variable and its recursive definition in the same equivalence class, 
whereas we use explicit recursive definitions and explicit equations between terms.

\subsection{Problem Formulation}
\label{sec:fo_problem_formulation}
We present the definition of a unification problem in this section, 
and then present in Section~\ref{sec:fo_preprocessing} a flattened form of the unification problem that the algorithm and the proofs assume, obtainable by preprocessing. 
With three syntactic entities, 
\emph{constructors} (written $c$, $d$, or $e$) and
\emph{unification metavariables} (written in capital letters $E$, $F$, $G$, $H$), 
\emph{recursion constants} (written in underlined letters $\recvar r, \recvar s, \recvar t$) \cite{Chen23fossacs}, possibly with subscripts,
a first-order \emph{concrete unification context} $\Delta_c$ is a system of equations $T \doteq T'$ together with definitions for
recursion constants that may occur in $T$. 
In term $c\, T_1\, \dots\, T_n$, $c$ is the \emph{head} and $T_1\, \dots\, T_n$ are the \emph{arguments} of $c$.
The grammar is shown as follows. It
enforces that 
recursive definitions are required to be contractive: $\recvar{r} =_d T$ means the head of $T$ must be a constructor.

\begin{center}
\begin{tabular}{lll}
  Concrete Unification Contexts & $\Delta_c ::= []\mid \Delta_c, T_1 \doteq T_2 \mid \Delta_c, \recvar r =_d c\, T_1\, \dots\, T_n $\\
  Terms & $ T ::=  c \, T_1\, \dots\, T_n \mid H \mid \recvar r $ \\
\end{tabular}
\end{center}

We now define the infinitary denotation of $T$ in a context $\Delta_c$ by depth $k$ observations of $M$. 
First, define $M_\bot$ to be first-order terms with the symbol $\bot$, and contractive and recursive unification metavariables (defined later in Section~\ref{sec:fo_preprocessing}).
\[M_\bot ::= c\, (M_\bot)_1\, \dots\,(M_\bot)_n \mid H\hcontra \mid H\hrec \mid \bot\]

We define \emph{definitional expansion} up to depth $k$ of a term $T$ into $M_\bot$ as the function $\m{exp}^{\Delta_c}\depth k(T) =M_\bot$, defined by lexicographic induction on $(k, T)$.
Since the parameter $\Delta_c$ remains unchanged throughout, we omit writing it to reduce visual clutter if it is not referenced.

\begin{center}
  \begin{tabular}{ll}
  $\m{exp}\depth{0}(T) = \bot$ \\
  $\m{exp}\depth{k+1}(c\, T_1\, \dots\, T_n) = c\, (\m{exp}\depth{k}(M_1))\, \dots \, (\m{exp}\depth{k}(M_n))$ \\
  $\m{exp}\depth{k+1}(H)  = H\hrec$ \\
  $\m{exp}^{\Delta_c}\depth{k+1}(\recvar r)  = \m{exp}\depth{k+1}(c\, T_1\, \dots\, T_n)$ if $\recvar r =_d c\, T_1\, \dots \, T_n \in \Delta_c$\\
  \end{tabular}
\end{center}

As an example, given a signature for conatural numbers and their simple equality, we are asked to find 
which number's double cosuccessor is $\mt{omega}$. 

\begin{verbatim}
conat : cotype.
cozero : conat.
cosucc : conat -> conat.

omega : conat = cosucc omega.
?- omega = (cosucc (cosucc H)).
\end{verbatim}

We may formulate the problem as follows, where $H$ is a fresh unification metavariable standing for the answer to our query.
\[\Delta_c = \{\recvar{omega} =_d \mi{cosucc}\, \recvar{omega}, \mi{cosucc} \,(\mi{cosucc}\, H) \doteq \recvar{omega}\}\]

We will not define unifiers for the concrete unification context, but the definition would look 
similar to the one for the unification context after the preprocessing phase defined next.
Eventually, we will find the following unifier for $\Delta_c$.
\[\Gamma_c = \{ H~\doteq~\recvar{omega},\recvar{omega} =_d \mi{cosucc}\, \recvar{omega}  \}\]

Notice that given two terms in a concrete context, their equality of definitional 
expansion up to depth $\omega$ is decidable \cite{Chen23fossacs,Huet98mscs}. The core idea is to carry out structural 
comparisons of the terms and memoizing intermediate equalities. This comparison terminates because the rationality 
of the terms ensures that only a finite number of intermediate equalities are possible.

\subsection{Preprocessing}
\label{sec:fo_preprocessing}

Terms are now divided into recursive terms and contractive terms.
A \emph{contractive term} is a term with a constructor as its head, and a \emph{recursive term} is a term with a recursion constant as its head.
The purpose of preprocessing is to put recursive definitions into 
shallow forms that are one level deep, meaning that
the arguments to a constructor must not be contractive terms and can only be recursive terms.
That means a term $c\, (d\, e)$ must be written down using recursive definitions: $c\, \recvar{r}, \recvar{r} =_d d\, \recvar{s}, \recvar{s} =_d\, e$ .
This greatly simplifies the termination proof of the 
unification algorithm here and for the higher-order case, which we eventually wish to develop.
Similarly, unification metavariables are 
divided into \emph{recursive unification metavariables} (with superscript $\hrec$), which may unify with only recursion constants, and 
\emph{contractive unification metavariables} (with superscript $\hcontra$) which may only unify contractive terms. 
We use the lower case letter $m$ to denote either $\symbcontra$ or $\symbrec$ and write $H^m$ to indicate a unification metavariable $H$ that is 
either contractive or recursive.
We also include a special symbol $\m{contra}$ 
for contradictory unification contexts that do not have a unifier. The unification context now only permits equations 
between two recursive terms or two contractive terms.
The grammar is as follows:

\begin{center}
\begin{tabular}{lll}
  Unification Contexts & $\Delta, \Gamma ::= []\mid \Delta, U_1 \doteq U_2 \mid \Delta, N_1 \doteq N_2 \mid \Delta, \recvar r =_d U \mid \Delta, \m{contra}$\\
  Contractive  Terms & $ U ::=  c \, N_1\, \dots\, N_n \mid H\hcontra $ \\
  Recursive Terms & $N ::=   \recvar r  \mid  H\hrec$ \\
\end{tabular}
\end{center}

As an example, we would like the concrete unification context $\Delta_c$ defined in the previous section
\[\Delta_c = \{\recvar{omega} =_d \mi{cosucc}\, \recvar{omega}, \mi{cosucc}\, (\mi{cosucc} H) \doteq \recvar{omega}\}\]
to be processed to the following unification context $\Delta$.
\[\Delta = \{\recvar{omega} =_d \mi{cosucc}\, \recvar{omega}, \recvar{s} =_d \mi{cosucc} \recvar{r}, \recvar{r} =_d\mi{cosucc} H\hrec, \recvar{s} \doteq \recvar{omega}\}\]

We define a preprocessing translation from concrete unification contexts to unification contexts.
We write $\Delta_c \rhd \Delta$ translation of $\Delta_c$ to $\Delta$, 
$T \rhd\hcontra U \diamond \Delta$ for translating a term $T$ into a contractive term $U$ with a new context $\Delta$, and 
$T \rhd\hrec N \diamond \Delta$ for translating a term $T$ into a recursive term $N$ with a new context $\Delta$. 
We treat a unification context as an unordered list and may write $\Delta_1,\Delta_2$ to join two contexts $\Delta_1$ and $\Delta_2$ 
with disjoint sets of recursion constants. If the set of recursion constants of $\Delta_1$ is not disjoint from the set of recursion constants of $\Delta_2$, we may 
consistently rename recursion constants in $\Delta_2$ such that $\Delta_1, \Delta_2$ is always defined.

\resetRN
\boxed{\Delta_c \rhd \Delta}
\begin{mathpar}
  \inferrule{ }{[] \rhd []}\RN

\inferrule{
\Delta_c \rhd \Delta_1 \\
 T_1 \rhd\hrec N_1 \diamond \Delta_2 \\
  T_2 \rhd\hrec N_2 \diamond \Delta_3}
  {\Delta_c, T_1 \doteq T_2 \rhd \Delta_1,  \Delta_2, \Delta_3, N_1 \doteq N_2}
  \RN \label{rule:_xxx_fefef}

  \inferrule{
\Delta_c \rhd \Delta_1 \\  c\, T_1\, \dots\, T_n \rhd\hcontra U \diamond \Delta_2
  }{
\Delta_c, \recvar r =_d  c\, T_1\, \dots \, T_n \rhd \Delta_1, \Delta_2, \recvar r =_d U
  }\RN
\end{mathpar}

\boxed{T \rhd\hrec N \diamond \Delta}
\begin{mathpar}
  \inferrule{
c\, T_1\, \dots\, T_n \rhd\hcontra U \diamond \Delta 
  }{
c\, T_1\, \dots\, T_n \rhd\hrec \recvar r \diamond (\Delta, \recvar r =_d U)
  }(\recvar{r} \text{ fresh})\RN
  \label{rule:fo_trans_into_rec}

  \inferrule{
  }{
    H \rhd H\hrec \diamond []
  }\RN

  \inferrule{
  }{
    \recvar r \rhd \recvar r \diamond []
  }\RN
\end{mathpar}

\boxed{T \rhd\hcontra U \diamond \Delta}
\begin{mathpar}
  \inferrule{
    \forall_{i, 1 \le i \le n}. T_i \rhd\hrec N_i \diamond \Delta_i
  }{
c\, T_1\, \dots\, T_n \rhd\hcontra c\, N_1\, \dots\, N_n \diamond (\Delta_1, \dots, \Delta_n)
  }\RN
  \label{rule:fo_trans_into_contra}

  (\text{No rules for $T = H$ or $T = \recvar r$})
\end{mathpar}

The informal intuition of the transformation rules is that we make transform both sides of a unification equation 
into recursive terms and transform the body of a recursive definition to a contractive term. 
Recursion constants are already recursive terms.
To transform any term 
with a constructor head to a contractive term, we transform the arguments into recursive terms (Rule (\ref*{rule:fo_trans_into_contra})).
To transform any non-recursive term into a recursive term, create a recursive definition that mimics the term (Rule (\ref*{rule:fo_trans_into_rec})).

We define the definitional expansion at depth $k$ for a recursive or a contractive term mutually recursively,
 $\m{exp}^{\Delta}\depth k(U) = M_\bot$ and $\m{exp}^{\Delta}\depth k(N) = M_\bot$. We take the liberty 
 to omit writing $\Delta$ if it 
 remains unchanged throughout and is not referenced.

 \begin{center}
 \begin{tabular}{ll}
 $\m{exp}\depth 0(U) = \bot$ \\
 $\m{exp}\depth{k+1}(H\hcontra) = H\hcontra$ \\
 $\m{exp}\depth{k+1}(c\,N_1\,\dots\,N_n  ) = c\, (\m{exp}\depth{k}(N_1))\, \dots\, (\m{exp}\depth{k}(N_n))$\\[1ex]
 $\m{exp}\depth 0(N) = \bot$ \\
 $\m{exp}\depth{k+1}(H\hrec) = H\hrec$ \\
 $\m{exp}^\Delta\depth{k+1}(\recvar r) = \m{exp}^\Delta\depth{k+1}(U)$ if $\recvar r =_d U \in \Delta$ \\
 \end{tabular}
\end{center}

The translation preserves the definitional expansion of arbitrary depth.

\begin{theorem}
  We have
  \begin{enumerate}
    \item 
If $T \rhd \hcontra U \diamond \Delta_2$
  and $\m{exp}^{\Delta_c}\depth k (\recvar s) = \m{exp}^{\Delta_1}\depth k (\recvar s)$ for all $s$ occurring in $T$,  
  then $\m{exp}^{\Delta_c}\depth k (T) = \m{exp}^{\Delta_1, \Delta_2}\depth k(U)$.
    \item 
If $T \rhd \hrec N \diamond \Delta_2$
  and $\m{exp}^{\Delta_c}\depth k (\recvar s) = \m{exp}^{\Delta_1}\depth k (\recvar s)$ for all $s$ occurring in $T$,  
  then $\m{exp}^{\Delta_c}\depth k (T) = \m{exp}^{\Delta_1, \Delta_2}\depth k(N)$.
  \end{enumerate}
\end{theorem}

\begin{proof}
  Simultaneous induction on $(k, T)$, where (2) may appeal to (1) without a decrease in size.

  We show the case for rules (\ref*{rule:fo_trans_into_rec})(\ref*{rule:fo_trans_into_contra}) as examples. 

  Rule (\ref*{rule:fo_trans_into_rec}):
  The premise $c\, T_1\, \dots\,T_n \rhd\hcontra U \diamond \Delta$ implies, by the induction hypothesis, that
  $\m{exp}^{\Delta_c}\depth k (c\, T_1\, \dots\,T_n ) = \m{exp}^{\Delta_1, \Delta}\depth k(U)$. 
  We have
  
  \begin{tabular}{ll}
$\m{exp}^{\Delta_c}\depth k (c\, T_1\, \dots\,T_n ) $ & \\
$ = \m{exp}^{\Delta_1, \Delta}\depth k ( U) $ & (by the induction hypothesis)\\
$ = \m{exp}^{\Delta_1, \Delta, \recvar r=_dU}\depth k ( U) $ & (because $\recvar r$ is fresh)\\
$ = \m{exp}^{\Delta_1, \Delta, \recvar r=_d U}\depth k( \recvar r) $ &  (by definition)\\
  \end{tabular}

  Rule (\ref*{rule:fo_trans_into_contra}): when $k = 0$, the result is trivial. When $k > 0$, 
  The premise $T_i \rhd\hrec N_i \diamond \Delta_i'$ implies, by the induction hypothesis, that
  $\m{exp}^{\Delta_c}\depth {k-1} (T_i) = \m{exp}^{\Delta_1, \Delta_i'}\depth {k-1}(N_i)$. 
  Let $\Delta' = \Delta_1', \dots, \Delta_n'$, we have
  $\m{exp}^{\Delta_c}\depth {k-1} (N_i) = \m{exp}^{\Delta_1, \Delta_i'}\depth {k-1}(N_i) = \m{exp}^{\Delta_1, \Delta'}\depth {k-1}(N_i)$, 
  because each $\Delta_i'$ may only contain fresh recursion constants.
  We have
  
  \begin{tabular}{ll}
$\m{exp}^{\Delta_c}\depth {k} (c\, T_1\, \dots\,T_n )$ & \\
$= c\, (\m{exp}^{\Delta_c}\depth {k-1}(T_1))\, \dots\, (\m{exp}^{\Delta_c}\depth {k-1}(T_n))$  & (by definition) \\
$= c\, (\m{exp}^{\Delta_1, \Delta'}\depth {k-1}(N_1))\, \dots\, (\m{exp}^{\Delta_1, \Delta'}\depth {k-1}(N_n))$  & (by the induction hypothesis) \\
$= \m{exp}^{\Delta_1, \Delta'}\depth k(c\, N_1\, \dots\, N_n)$  & (by definition) \\
  \end{tabular}
\end{proof}

\begin{corollary}
  If $\Delta_c \rhd \Delta$, then every equation in $\Delta_c$ corresponds to an equation in $\Delta$ with equal 
  definitional denotation, and every recursive definition in $\Delta_c$ corresponds to a recursive definition in $\Delta$.
\end{corollary}

\begin{proof}
  Directly by structural induction over $\Delta_c \rhd \Delta$.
\end{proof}

It is worth noting that we have assumed all concrete unification metavariables $H$ are recursive, in the sense that 
they may unify with a recursion constant. In practice, implementations may want to use the preprocessed forms 
directly. The concrete form and the translation procedure merely serve as a mechanism to parse the user's input and as a 
formal explanation of the flattened definitions. 

We take the flattened unification context as the ``canonical representation'' for a unification problem 
from now on, and we may use the syntax category $M$ for either $U$ or $N$.
We use $\m{defs}(\Delta)$ and $\m{eqs}(\Delta)$ to denote the list of recursive definitions and equations of $\Delta$ respectively.
Definitional expansion $\m{exp}$ does not depend 
on unification equations but only on recursive definitions, and thus, we have 
$\m{exp}^{\Delta}\depth k (M)= \m{exp}^{\m{defs}(\Delta)}\depth{k}(M)$, for all $\Delta$, $k$, and $M$.

\subsection{Term Equality and Unifiers}
\label{sec:fo_term_equality_and_unifiers}

Two terms are equal in a unification context if they have the same definitional expansion, i.e., given $M \doteq M'$ in 
$\Delta$, we say that $M$ is equal to $M'$ (and thus the equation $M \doteq M'$ \emph{holds}) if $\m{exp}^\Delta\depth{k}(M) = \m{exp}^\Delta\depth k(M')$ for all $k$.
We say that a unification context is \emph{contradiction-free} if $\m{contra}$ is not present in the context.

A (simultaneous) substitution is usually understood as a mapping from unification metavariables to terms.
In the case of circular terms, the substitutions may carry recursive definitions. We choose to define substitutions 
as unification contexts of special forms, where the left-hand sides of all unification equations are unification metavariables, and 
the corresponding right-hand sides are their values. We write $\Gamma$ for substitutions and $\Delta$ for ordinary unification contexts.
A \emph{substitution}  is a contradiction-free unification context
where the left-hand side of each unification equation is a unique unification metavariable.
The set of unification metavariables that occur on the left-hand sides 
of a substitution $\Gamma$ is called the \emph{domain} of the substitution and is written $\m{dom}(\Gamma)$.
If a substitution contains an equation $H^m \doteq M$, 
we say that $M$ is the \emph{value} of $H^m$ in $\Gamma$.
Two substitutions are equal if they have the same domain, and 
the definitional expansions of the values of each unification metavariable in their domain are equal, i.e. 
$\Gamma = \Gamma'$ if 
$\m{dom}(\Gamma) = \m{dom}(\Gamma')$, 
and for all $H^m \in \m{dom}(\Gamma)$, $\m{exp}^\Gamma\depth{k}(H^m[\Gamma]) = \m{exp}^{\Gamma'}\depth{k}(H^m[\Gamma])$, 
where $H^m[\Gamma]$ is the value of $H^m$ in $\Gamma$, obtained by the substitution operation that will be defined.

As an example, $\Gamma$ and $\Gamma'$ below are substitutions with the domain $\{H\hrec\}$.

\begin{tabular}{lll}
$\Gamma$&$ = \{H\hrec \doteq \recvar{omega}, $&$\recvar{omega} =_d \mi{cosucc}\, \recvar{omega}\}$\\
$\Gamma'$&$ = \{H\hrec \doteq \recvar{s}, $&$\recvar{omega} =_d \mi{cosucc}\, \recvar{omega}, \recvar{s} =_d \mi{cosucc} \recvar{omega}\}$\\
\end{tabular}

Moreover, $\Gamma = \Gamma'$ because the expansions of every unification metavariable in the domain are equal: $H\hrec$ expands to $\mi{cosucc}\,  (\mi{cosucc} \dots)$.

We emphasize that in a substitution, unification metavariables occurring on the right-hand sides of unification equations and in 
recursive definitions are free. Thus, the substitution $\Gamma''$ below has $H\hrec$ in the recursive definition free, and 
$\Gamma''$ is not equal to $\Gamma$ defined above.

\begin{tabular}{ll}
$\Gamma'' = \{H\hrec \doteq \recvar{omega}, \recvar{omega} =_d \mi{cosucc}\, H\hrec\}$\\
\end{tabular}

 We write $U[\Gamma]$ and $N[\Gamma]$ for applying the substitution to terms. They are defined in obvious ways.

 \begin{tabular}{ll}
  $(c\, N_1\, \dots\, N_n)[\Gamma] = c\, (N_1[\Gamma])\, \dots\,(N_n[\Gamma])$ \\
  $(H\hcontra)[\Gamma] = \begin{cases}
    U'            & \text{ if } H \hcontra \doteq U' \in \m{eqs}(\Gamma)\\
    H\hcontra  & \text{ otherwise} \\
  \end{cases}$ \\[1ex]
  $(\recvar r)[\Gamma] = \recvar r$ \\
  $(H\hrec)[\Gamma] = \begin{cases}
    N'    & \text{ if } H \hrec \doteq N' \in \m{eqs}(\Gamma)\\
    H\hrec  & \text{ otherwise} \\
  \end{cases}$ \\
 \end{tabular}

The application of substitution $\Gamma$ to a \emph{unification context} $\Delta$ is denoted $\Delta[\Gamma]$, which replaces occurrences of 
 unification metavariables in $\Gamma$ by their values in $\Delta$, while combining all recursive definitions and performing recursion constant renaming as necessary.
 \[\Delta[\Gamma] = \m{defs}(\Gamma), \{M[\Gamma] \doteq M'[\Gamma] \mid M \doteq M' \in \m{eqs}(\Delta)\},  \{\recvar r =_d U[\Gamma] \mid \recvar r =_d U \in \m{defs}(\Delta)\}\]
The application of a substitution $\Gamma_2$ to another \emph{substitution} is $\Gamma_1$  is denoted $ \Gamma_1[\Gamma_2]$, 
and it replaces the occurrences of unification metavariables in the right-hand sides and recursive definitions of $\Gamma_1$ by 
their values in $\Gamma_2$ and combine all recursive definitions, performing recursion constant renaming as necessary.
\[\Gamma_1[\Gamma_2] = \m{defs}(\Gamma_2),  \{H^m \doteq M'[\Gamma_2] \mid H^m \doteq M' \in \m{eqs}(\Gamma_1)\},  \{\recvar r =_d U[\Gamma_2] \mid \recvar r =_d U \in \m{defs}(\Gamma_1)\}\]
The composition of substitutions is denoted $\Gamma_1 \circ \Gamma_2$ (applying $\Gamma_1$ and then applying $\Gamma_2$), 
is defined to be $\Gamma_1[\Gamma_2]$ plus any additional substitutions in $\Gamma_2$.
\[ \Gamma_1 \circ \Gamma_2 = (\Gamma_1[\Gamma_2]), \{H^m \doteq M \mid H^m \doteq M \in \m{eqs}(\Gamma_2) \land H^m \notin \m{dom}(\Gamma_1)\} \]

Let $UV(\Delta)$ denote the set of all unification metavariables that occur in $\Delta$, 
a \emph{unifier} for a contradiction-free unification context $\Delta$ is a substitution $\Gamma$ such that $UV(\Delta) = \m{dom}(\Gamma)$, and every equation in $\Delta[\Gamma]$ holds. 
A unification context $\Delta$ with $\m{contra} \in \Delta$ has no unifiers.
A unifier $\Gamma_1$ is \emph{more general} than another unifier $\Gamma_2$  if
there is a substitution $\Gamma'$ such that $\Gamma_1 \circ \Gamma' = \Gamma_2$.

As an example, given $\Gamma$ and $\Delta$ defined below, $\Gamma$ is a unifier of $\Delta$, because every 
equation in $\Delta[\Gamma]$ holds. Notice that when carrying out the substitution, the duplicate recursion constant $\recvar{omega}$ in $\Gamma$ 
is renamed to $\recvar{t}$. The major changes are highlighted in \blue{blue}.

\begin{multicols}{3}
\begin{tabular}{ll}
 $\Gamma $ $= \{ $ \\
$\  H\hrec \doteq \recvar{omega}, $\\
\\
\\
$\  \recvar{omega} =_d \mi{cosucc}\, \recvar{omega}$ \\
$\}$\\
\end{tabular}

\columnbreak
\begin{tabular}{ll}
$\Delta = \{ $\\
$\ \recvar{omega} =_d \mi{cosucc}\, \recvar{omega}, $\\
$\ \recvar{s} =_d \mi{cosucc} \recvar{r}, $\\
$\ \recvar{r} =_d\mi{cosucc} H\hrec, $\\
\\
$\ \recvar{s} \doteq \recvar{omega}$ \\
$\}$\\
\end{tabular}

\columnbreak
\begin{tabular}{ll}
$\Delta[\Gamma] =  \{$\\
$ \ \recvar{omega} =_d \mi{cosucc}\, \recvar{omega}, $\\
$ \ \recvar{s} =_d \mi{cosucc} \recvar{r},$\\ 
$ \ \recvar{r} =_d\mi{cosucc}\, \blue{\recvar{t}}, $\\
$ \ \blue{\recvar{t}} =_d \mi{cosucc}\, \blue{\recvar{t}}, $\\
$ \ \recvar{s} \doteq \recvar{omega}$ \\
$ \}$\\
\end{tabular}
\end{multicols}
  
\subsection{The Unification Algorithm}
\label{sec:fo_algorithm}

We saturate the unification context $\Delta$ using the rules defined below.
If all the premises of a rule are present in the context, we add the rule's conclusion to the context.
The algorithm terminates when no new equations or recursive definitions can be added to the context.
The goal of the rules is to ensure that in a saturated unification context, 
either $\m{contra}$ is present, indicating there is no unifier, 
or there is an equation between each unification metavariable and its value in a unifier.

\resetRN
\noindent\rule{\textwidth}{1pt}
Structural Rules:
\begin{mathpar}
\inferrule[\RNwithLabel{SIMP-F}\label{rule:fo_constructor_clash}]{
  c\,  N_1\, \dots\, N_n \doteq d\, N_1' \, \dots\, N_n'
}{\m{contra}}(c \ne d)

\inferrule[\RNwithLabel{SIMP}\label{rule:fo_structural}]{
  c\, N_1\, \dots N_n \doteq
  c\, N_1'\, \dots N_n'
}{
  N_1 \doteq N_1', \dots, N_n \doteq N_n'
}

\end{mathpar}

\noindent\rule{\textwidth}{1pt}
Expansion, Symmetry, and Transitivity
\begin{mathpar}
  \inferrule[\RNwithLabel{R-EXP}\label{rule:fo_rec_expand}]{
    \recvar{r}\doteq \recvar{s}\\
    \recvar{r} =_d   U_1 \\
    \recvar{s} =_d   U_2 \\
  }{
     U_1\doteq U_2
  }

\inferrule[\RNwithLabel{U-SYM}\label{rule:fo_sym_contra}]{
  U \doteq U'
}{U' \doteq U}

\inferrule[\RNwithLabel{U-TRANS}\label{rule:fo_trans_contra}]{
  U_1 \doteq U_2 \\ U_2 \doteq U_3 
}{U_1 \doteq U_3}

\inferrule[\RNwithLabel{N-SYM}\label{rule:fo_sym_rec}]{
  N \doteq N'
}{N' \doteq N}

\inferrule[\RNwithLabel{N-TRANS}\label{rule:fo_trans_rec}]{
  N_1 \doteq N_2 \\ N_2 \doteq N_3 
}{N_1 \doteq N_3}

\end{mathpar}
\noindent\rule{\textwidth}{1pt}

We give an example of the ways the algorithm operates on our previous example.
We label each equation with a number and use $\Delta_i$ to refer to the set of equations and definitions $(1) - (i)$.
For example, our example $\Delta$ is denoted $\Delta_4$, consisting of equations and definitions $(1) - (4)$.
At each step, we show some additional equations and definitions
and the ways they are obtained. We only show the first few important steps and the rest will be only symmetry and transitivity.

\begin{tabular}{ll}
$(1)\, \recvar{omega} =_d \mi{cosucc}\, \recvar{omega}$ & given \\
$(2)\, \recvar{s} =_d \mi{cosucc} \recvar{r}$ \\
$(3)\, \recvar{r} =_d\mi{cosucc} H\hrec$ \\
$(4)\, \recvar{s} \doteq \recvar{omega} $\\
$(5)\,  \mi{cosucc}\, \recvar r \doteq \mi{cosucc}\, \recvar{omega}$ & by Rule (\ref*{rule:fo_rec_expand}) on $(4)$, $(2)$ and $(1)$ \\
$(6)\, \recvar{r} \doteq \recvar{omega}$ &by Rule (\ref*{rule:fo_structural}) on $(5)$ \\
$(7)\, \mi{cosucc}\, H\hrec \doteq \mi{cosucc}\, \recvar{omega}$ &by Rule (\ref*{rule:fo_rec_expand}) on $(6)$, $(3)$ and $(1)$ \\
$(8)\, \mi{H\hrec} \doteq \recvar{omega}$ &by Rule (\ref*{rule:fo_structural}) on $(7)$ \\
$(9)\, \dots$ & by Rules (\ref*{rule:fo_sym_contra})(\ref*{rule:fo_trans_contra})(\ref*{rule:fo_sym_rec})(\ref*{rule:fo_trans_rec})
\end{tabular}

\subsection{Saturated Unification Contexts}
We now describe how a unifier may be constructed from a saturated contradiction-free unification context.
Given a unification context $\Delta$, we say that a unification metavariable $H\hcontra$ 
is \emph{resolved} if there is an equation of the form $H \hcontra \doteq c\, N_1\, \dots\, N_n$ or $c\, N_1\, \dots\, N_n \doteq H\hcontra$, and 
$c\, N_1\, \dots \, N_n$ is called a resolution of $H\hcontra$.
Similarly,  we say that a unification metavariable $H\hrec$ 
is \emph{resolved} if there is an equation of the form $H\hrec \doteq \recvar r$ or $\recvar r \doteq H\hrec$, and 
$\recvar r$ is called a resolution of $H\hrec$.
In a unification context, every unification metavariable is either resolved or unresolved. 
There may be multiple resolutions for each resolved unification metavariable, we pick a unique resolution for each unification metavariable. 
The choice of resolution is not important, because  every resolution will be equal modulo definitional expansion in a saturated contradiction-free context.
Unresolved unification metavariables 
form an equivalence class equated by $\doteq$, and we pick a unique representative variable for each class.
We construct the substitution $\Gamma = \m{unif}(\Delta)$ for a contradiction-free context $\Delta$ as follows.
\begin{enumerate}
  \item Start with $\Gamma$ containing all recursive definitions of $\Delta$.
  \item For each resolved unification metavariable in $UV(\Delta)$, add to $\Gamma$ the unification metavariable and its resolution.
  \item For each unresolved unification metavariable in $UV(\Delta)$, add to $\Gamma$ the unification metavariable and the representative unification metavariable for its equivalence class.
  \item Replace the occurrences of resolved unification metavariables in the right-hand sides and recursive definitions of $\Gamma$ with their resolutions, 
  and replace the occurrences of unresolved unification metavariables in the right-hand sides and recursive definitions of $\Gamma$ with their representative unification metavariables.
  Repeat this step until all unification metavariables in the right-hand sides and recursive definitions are representative unification metavariables for some 
  equivalence class of unresolved unification metavariables. 
\end{enumerate}
We will later show that if $\Delta$ is a saturated contradiction-free unification context, then $\Gamma = \m{unif}(\Delta)$ is a unifier for $\Delta$.
As an example, we show how the unifier for $\Delta_8$ (equations and definitions $(1)-(8)$ defined above) can be constructed.
The main differences in each step are highlighted in \blue{blue}.

\begin{enumerate}
  \item Initialize $\Gamma_1$ to all recursive definitions of $\Delta_8$.

  \item Since $H\hrec$ is resolved, we add its resolution to get $\Gamma_2$.

\item There is no unresolved unification metavariable, we skip step (3).

\item Replace occurrences of resolved unification metavariables with their resolutions to get $\Gamma_3$.

\begin{multicols}{3}
\begin{tabular}{ll}
$\Gamma_1 = \{$ \\
$\ \recvar{omega} =_d \mi{cosucc}\, \recvar{omega}, $\\
$\ \recvar{s} =_d \mi{cosucc} \recvar{r}, $ \\
$\ \recvar{r} =_d\mi{cosucc} H\hrec $\\
$\}$ \\
\end{tabular}

\columnbreak
\begin{tabular}{ll}
$\Gamma_2 = \{  $\\
$\ \recvar{omega} =_d \mi{cosucc}\, \recvar{omega}, $\\
$\ \recvar{s} =_d \mi{cosucc} \recvar{r}, $\\
$\ \recvar{r} =_d\mi{cosucc} H\hrec, $\\
$\ \blue{H\hrec \doteq \recvar{omega}}$\\
$\}$ \\
\end{tabular}

\columnbreak
\begin{tabular}{ll}
$\Gamma_3 = \{ $\\
$\ \recvar{omega} =_d \mi{cosucc}\, \recvar{omega}, $\\
$\ \recvar{s} =_d \mi{cosucc} \recvar{r}, $\\
$\ \recvar{r} =_d\mi{cosucc}\, \blue{\recvar{omega}}, $\\
$\ H\hrec~\doteq~\recvar{omega}$\\
$\}$ \\
\end{tabular}

\end{multicols}

\item Note that we may remove unused recursive definitions ($s, r$) to get an equivalent substitution $\Gamma_4$.

$\Gamma_4 = \{ \recvar{omega} =_d \mi{cosucc}\, \recvar{omega},  H\hrec~\doteq~\recvar{omega}\}$

\end{enumerate}

\subsection{Correctness of the Algorithm}
\label{sec:fo_metatheory}

We want to show that
given a unification context $\Delta_1$, it has a finitary saturation sequence 
\[\Delta_1 \to \Delta_2 \to \dots \to \Delta_n\]where
$\Delta_n$ is a saturated unification context that has $\m{unif}(\Delta_n)$ as its most general unifier.
Moreover,
unifiers are preserved between
$\Delta_{i}$ and $\Delta_{i+1}$.
  Then, the most general unifiers of $\Delta_n$ are 
  the most general unifiers of $\Delta_1$.
Concretely, we want to show three things: 
\begin{enumerate}[(1)]
  \item \textbf{Correspondence}. At each step of the algorithm, the most general unifier of the context before corresponds to the most general unifier of the context after (Theorem~\ref{thm:fo_correspondence}).
  \item \textbf{Termination}. Any unification context always saturates in a finite number of steps (Theorem~\ref{thm:fo_termination}).
  \item \textbf{Correctness}. The unifier for a saturated unification context $\m{unif}(\Delta_n)$ is actually the most general unifier (Theorem~\ref{thm:fo_correctness_of_unifiers}).
\end{enumerate}

\begin{lemma}
  \label{thm:fo_unifier_subset}
  Let $\Delta'$ be a unification context, and let $\Delta$  have the  same set of recursive definitions and 
  unification metavariables, but fewer unification equations than $\Delta'$, i.e.,
  $\m{eqs}(\Delta) \subseteq \m{eqs}(\Delta')$, $\m{defs}(\Delta) = \m{defs} (\Delta')$, $UV(\Delta) = UV(\Delta')$,
   then any unifier of $\Delta'$ is 
  a unifier of $\Delta$.
\end{lemma}
\begin{proof}
  
  Let $\Gamma$ be a unifier of $\Delta'$,  all unification equations of $\Delta'[\Gamma] $ hold.
  Take any $M \doteq M' \in \Delta$, we know $\m{exp}^{\Delta' [\Gamma]}\depth k (M[\Gamma])= \m{exp}^{\Delta' [\Gamma]}\depth k(M'[\Gamma])$, 
  we have $UV(\Delta) = UV(\Delta')$, and it suffices to show $\m{exp}^{\Delta [\Gamma]}\depth k (M[\Gamma])= \m{exp}^{\Delta [\Gamma]}\depth k(M'[\Gamma])$ 
  by showing $\m{exp}^{\Delta [\Gamma]}\depth k (M[\Gamma])= \m{exp}^{\Delta' [\Gamma]}\depth k(M[\Gamma])$.  
  But since definitional expansions only depend on recursive definitions, we have 

  \begin{tabular}{ll}
 $\m{exp}^{\Delta [\Gamma]}\depth k (M[\Gamma])$ \\
 $= \m{exp}^{\m{defs}(\Delta [\Gamma])}\depth k(M[\Gamma])$ & (expansion only definition only depends on definitions of $\Delta[\Gamma]$) \\
 $= \m{exp}^{\m{defs}(\Delta' [\Gamma])}\depth k(M[\Gamma])$ & ($M$ can only depend on recursion constants occurring in $\Delta$)\\
 $= \m{exp}^{\Delta' [\Gamma]}\depth k(M[\Gamma]) $ &(expansion only definition only depends on definitions of $\Delta'[\Gamma]$) \\
  \end{tabular}

\end{proof}

\begin{lemma}
  \label{thm:fo_unifier_superset}
  If $\Gamma$ is a unifier for $\Delta$, then $\Gamma$ is a unifier for $\Delta'$
  where $\Delta'$ has all recursive definitions of $\Delta$ and additional true equations 
  $M \doteq M'$  in the sense that $\m{exp}^{\Delta[\Gamma]}\depth k (M[\Gamma]) = \m{exp}^{\Delta[\Gamma]}\depth k (M'[\Gamma])$.
\end{lemma}
\begin{proof}
  Because definitional expansion depends only on recursive definitions but not unification equations, 
  we have 
  $\m{exp}^{\Delta [\Gamma]} \depth k (M) = \m{exp}^{\Delta' [\Gamma]} \depth k (M)$ for all $k$ and $M$.
\end{proof}

\begin{theorem}
  [Correspondence]
  \label{thm:fo_correspondence}
  If $\Delta'$ is obtained from $\Delta$ by applying one of the rules, 
  then the unifiers of $\Delta'$
  and the unifiers of $\Delta$ coincide.
\end{theorem}
\begin{proof}
  We analyze each rule.

  Case (\ref*{rule:fo_constructor_clash}), both $\Delta$ and $\Delta'$ have no unifiers. 

  Case (\ref*{rule:fo_structural}), it's easy to check that any unifier $\Gamma$ of $\Delta'$ is a unifier of $\Delta$ by Lemma \ref{thm:fo_unifier_subset}.
  Now suppose $\Gamma$ is a unifier of $\Delta$, we want to show that $\Gamma$ is a unifier for $\Delta'$. 
  The additional equations $M_i \doteq M_i'$ in $\Delta'$ satisfy $\m{exp}^{\Delta [\Gamma]}\depth k(M_i[\Gamma]) = \m{exp}^{\Delta [\Gamma]}\depth k (M_i'[\Gamma])$, 
  and the rest follows by Lemma~\ref{thm:fo_unifier_superset}.

  The rest of the cases are similar to Case (\ref*{rule:fo_structural}).

\end{proof}

\begin{theorem}
  [Termination]
  \label{thm:fo_termination}
  The saturation algorithm always terminates.
\end{theorem}
\begin{proof}
  We observe that all terms in an equation are built up from recursion constants, unification metavariables, and 
  constants, and all terms have finite depth (due to the grammar) and finite width (the maximum width is preserved 
  by the algorithm). There can be only finitely many equations 
  given a bounded number of recursion constants, constructors, and unification metavariables,
  and there are no rules that create additional recursion constants, constructors, or unification metavariables.
  
\end{proof}

\begin{theorem}[Correctness of Unifiers]
  \label{thm:fo_correctness_of_unifiers}
  Given any saturated contradiction-free unification context $\Delta$, let $\Gamma = \m{unif}(\Delta)$, then 
  $\Gamma$ is a unifier for $\Delta$. Moreover, it is the most general unifier.
\end{theorem}
\begin{proof}
  The proof is broken into two parts. The first part is
  to show that $\Gamma$ is a unifier, and the second part is to show that 
  $\Gamma$ is the most general unifier.

  \textbf{(Part 1)}
  To show $\Gamma$ is a unifier, 
  we need to show that $\m{dom}(\Gamma) = UV(\Delta)$, which is true by definition, and that every equation in $\Delta[\Gamma]$  holds.
  We show the following two claims simultaneously by induction on $k$, where claim (2) may refer to claim (1)
  without decreasing $k$.
  \begin{enumerate}
    \item For all $U_1 \doteq U_2$ in $\Delta$, $\m{exp}^{\Delta[\Gamma]}\depth{k}(U_1[\Gamma]) = \m{exp}^{\Delta[\Gamma]}\depth{k}(U_2[\Gamma])$.
    \item For all $N_1 \doteq N_2$ in $\Delta$, $\m{exp}^{\Delta[\Gamma]}\depth{k}(N_1[\Gamma]) = \m{exp}^{\Delta[\Gamma]}\depth{k}(N_2[\Gamma])$.
  \end{enumerate}
  Both claims are trivial when $k = 0$. Consider the case when $k > 0$, we show (1) and (2) by case 
  analysis on the structure of $U_1 \doteq U_2$ and $N_1 \doteq N_2$. 

  \begin{enumerate}[(a)]
  \item Both $U_1$ and $U_2$ have constructors as their heads. Since $\m{contra} \notin \Delta$, they must 
   have identical constructor heads. Now let $U_1 = c\, N_1\, \dots \, N_n$ and $U_2 = c\, N_1'\, \dots\, N_n'$.
    Since $\Delta$ is saturated, we have $N_i \doteq N_i'$ for all $1 \le i \le n$. The result 
    then follows from the fact that each $N_i$ and $N_i'$ 
    have equal definitional expansion up to depth $k-1$ by induction hypothesis. 

    \item Both $U_1$ and $U_2$ are contractive unification metavariables. Due to saturation, either both are unresolved,
    and the result follows because they would be in the same equivalence class and thus have the same representative unification metavariable,
     or both are resolved. 
    If they have a unique resolution $U$, then we have $\m{exp}\depth{k}(U_1)  = \m{exp}\depth{k}(U) = \m{exp}\depth{k}(U_2)$.
    If they have multiple resolutions and one of the resolutions is $U$, saturation 
    guarantees that there is an equation between every resolution. 
    Rule (\ref*{rule:fo_structural}) ensures that the equations between children of the head constructors are in $\Delta$,
    and the two terms would be equal by IH, using a similar technique as the case (a).

    \item One of $U_1$ and $U_2$ is a unification metavariable, and the other has a constructor as its head. 
    Obviously this is a resolution equation, and it suffices to show that all other resolutions have equal definitional expansions up to depth $k$, 
    which follows from saturation and the case (a).

    \item Both $N_1$ and $N_2$ are recursion constants. Let $N_1 = \recvar r$, where $\recvar r =_d U_1 \in \Delta$ and $N_2 = \recvar s$, where $\recvar s =_d U_2 \in \Delta$.
    Since $\Delta$ is saturated, $U_1 \doteq U_2 \in \Delta$, and by IH (i.e. case (a) above), $\m{exp}\depth{k}^{\Delta[\Gamma]}(U_1[\Gamma]) = \m{exp}\depth{k}^{\Delta[\Gamma]}(U_2[\Gamma])$, and

    \begin{tabular}{ll}
    $\m{exp}\depth{k}^{\Delta[\Gamma]}(\recvar r[\Gamma])$ & \\
     $= \m{exp}\depth{k}^{\Delta[\Gamma]}(\recvar r)$ & (since $\recvar{r}[\Gamma] = \recvar{r}$)\\
      $= \m{exp}\depth{k}^{\Delta[\Gamma]}(U_1[\Gamma])$ & (by the definition of $\Delta[\Gamma]$ and definitional expansion)\\
    $= \m{exp}\depth{k}^{\Delta[\Gamma]}(U_2[\Gamma])$ & (shown)\\
    $= \m{exp}\depth{k}^{\Delta[\Gamma]}(\recvar s) $ & (by the definition of $\Delta[\Gamma]$ and definitional expansion)\\
    $= \m{exp}\depth{k}^{\Delta[\Gamma]}(\recvar s[\Gamma]) $ &(since $\recvar{s}[\Gamma] = \recvar{s}$) \\
    \end{tabular}

    \item The case when either or both of $N_1$ and $N_2$ are recursive unification metavariables are exactly analogous 
    to cases (b) and (c).
    
  \end{enumerate}

  \textbf{(Part 2)}
  To show $\Gamma$ is the most general unifier, given any other unifier $\Gamma_2$ of $\Delta$, it suffices to 
  construct a unifier $\Gamma_1$ such that $\Gamma \circ \Gamma_1 =\Gamma_2$. But 
  the construction of $\Gamma_1$ is easy: $\Gamma_2$ must map resolved unification metavariables analogously as 
  $\Gamma$ (otherwise a contradiction will arise), and it may choose to map equivalence classes of unresolved unification metavariables freely. $\Gamma_1$ simply records how unresolved unification metavariables are mapped in $\Gamma_2$.

\end{proof}

\section{Higher-Order Pattern Unification}
\label{sec:ho_case}

In this section, we extend the algorithm to a higher-order setting. In particular, both recursive definitions and 
unification metavariables are of higher-order types: they may be applied to pattern variables. 
The main technical challenge of the higher-order case is to handle the scoping of variables in the presence of recursive 
definitions.  We copy Miller's higher-order pattern unification algorithm \citet{Miller91jlc} for non-recursive cases 
and handle the recursive definitions by delegating the scoping to unification metavariables (i.e. Rule (\ref*{rule:ho_rec_pruning}) in Section~\ref{sec:algorithm}).
In terms of presentation, recursion constants and unification metavariables are always applied to pattern variables.
We need to update the definitions to take variable renaming into account.

\subsection{Problem Formulation}
\label{sec:problem_formulation}

We now give a similar development by allowing recursion constants and unification metavariables to carry patterns \cite{Miller91jlc}.
A \emph{pattern} is a list of pairwise distinct bound variables (written $x$, $y$, or $z$), and 
the \emph{pattern restriction} ensures that
a recursion constant or a unification metavariable may only be applied to
a pattern. 
Here's an example of a higher-order pattern unification problem (without recursive definitions).
\[\lambda x. \, \lambda y.\, \lambda z.\, c \, (F\, x\, y) \doteq 
\lambda x. \, \lambda y.\, \lambda z.\, c\, (G\, y\, z) 
\]
A variable may not appear free in a unifier.
For instance, the substitution $F \, x\, y=x, G \, y\, z=  x$ (i.e. $F = \lambda x.\, \lambda y.\, x, G = \lambda y.\, \lambda z.\, x$) is not a unifier because $x$ is free in the substitution of $G$
but the substitution  $F  \, x\, y=  d, G \, y\, z= d$ is a unifier.

Regular B\"ohm trees \cite{Huet98mscs}, subsequently termed higher-order rational terms, provide a natural model 
for higher-order terms.
As with the first-order case, our use of a context containing recursive definitions for recursion constants
follows the design of CoLF \cite{Chen23fossacs}. While CoLF allows repetitions of bound variables in 
the arguments to recursion constants, we disallow them in the setting of unification to ensure that most general unifiers exist.
This is not a restriction in practice, because 
 any recursive definition with repetitive arguments 
  can be rewritten to definitions within the pattern fragment,  as observed by \citet{Huet98mscs}.
  For example, if we have non-pattern appears as arguments to a recursion constant $\recvar{r}\, x\, x$, with $\recvar{r} =_d \lambda y.\, \lambda z.\, T$, 
  we can always create a fresh recursion constant $\recvar{t}$ that mimics $\recvar{r}$, i.e., $\recvar{t}\, x$ with $\recvar{t} =_d \lambda w.\, [w,w/y,z]T$.
We assume that every unification metavariable, recursion constant, constructor, or variable is assigned a simple type, and 
terms are always written in $\beta$-normal-$\eta$-long forms, except that 
arguments to recursion constants and unification metavariables are written in non-expanded forms.
We also assume that $\lambda$-bound variables may undergo $\alpha$-renaming.
Here's the grammar for the unification problem in concrete syntax.
\begin{center}
  \begin{tabular}{lll}
    Concrete Unification Contexts & $\Delta_c ::= []\mid \Delta_c, T_1 \doteq T_2 \mid \Delta_c, \recvar r =_d \lambda x_1.\dots\lambda x_l.\,  h\, T_1\, \dots\, T_n $\\
    Terms & $ T ::=  \lambda x_1. \dots. \lambda x_l. \, h \, T_1\, \dots\, T_n \mid \lambda x_1. \dots. \lambda x_l. \, H \, y_1\, \dots\, y_n $ \\
     & $\qquad\mid  \lambda x_1. \dots. \lambda x_l. \, \recvar{r} \, y_1\, \dots\, y_n $
    \\
    Constructor or Variable Heads & $ h ::=  c \mid x $ \\
  \end{tabular}
  \end{center}

  To avoid visual clutter when writing down a list of terms, we adopt the following conventions of using overlines to represent a list of terms.
  \begin{enumerate}
    \item A list of variables
  $\bar{x}$ means $x_1, \dots, x_l$ that are pairwise distinct.
  \item A list of variables appearing in a binder position means iterative abstractions. For example,
  $\lambda\bar{x}$ means $\lambda x_1. \dots. \lambda x_l.$.
    \item A list of variables in an application means iterative applications.
  For example, $c\, \bar{x}$ means $c\, x_1 \,\dots\, x_n$. Similarly, a list of terms in an application position means iterative applications.
   For example, $h\, \overline{N}$ means $h\, N_1\,\dots\, N_n$.
  
   \item The notation $[\bar y/\bar x]M$ denotes the simultaneous renaming of variables, substituting $\bar y$ for $\bar x$ in $M$.
  \end{enumerate}

  With the new abbreviation notation, the grammar for the concrete syntax for a unification problem may be written as the following.
\begin{center}
  \begin{tabular}{lll}
    Concrete Unification Contexts & $\Delta_c ::= []\mid \Delta_c, T_1 \doteq T_2 \mid \Delta_c, \recvar r =_d \lambda \bar x.\,  h\,  \overline{T} $\\
    Terms & $ T ::=  \lambda \bar x.\, h\, \overline T \mid \lambda \bar x. \, H \, \bar y \mid \lambda \bar x.\, \recvar r\, \bar y$ \\
    Constructor or Variable Heads & $ h ::=  c \mid x $ \\
  \end{tabular}
  \end{center}

The grammar enforces that the definition for a recursion constant is required to 
be contractive: it has a variable or a constructor for its head. We use $FV(T)$ to denote 
the set of free variables in $T$.
We require all recursive definitions to be closed in the sense that $r =_d \lambda \bar x.\,  h\,  \overline{T}  \in \Delta_c$ implies that $FV(h\, \overline{T}) \subseteq \bar x$.

As with the first-order case, we define the infinitary denotation of $T$ in a context $\Delta_c$ by depth $k$ observations of $M$. 
Now $M_\bot$ includes $\lambda$-bindings and variables. 
\[M_\bot ::=\lambda \bar x.\, y\, \overline{M_\bot} \mid\lambda \bar x.\, c\, \overline{M_\bot} \mid \lambda \bar x.\, H\hcontra\, \overline{M_\bot}  \mid \lambda \bar x.\, H\hrec\, \overline{M_\bot}  \mid \bot\]

We define \emph{definitional expansion} up to depth $k$ of a term $T$ into $M_\bot$ as the function $\m{exp}^{\Delta_c}\depth k(T) =M_\bot$, defined by lexicographic induction on $(k, T)$.
We omit $\Delta_c$ to reduce visual clutter if it is not referenced.

\begin{center}
  \begin{tabular}{ll}
  $\m{exp}\depth{0}(T) = \bot$ \\
  $\m{exp}\depth{k+1}(\lambda \bar x.\,h\, T_1 \, \dots\, T_n) =\lambda \bar x.\, h\, (\m{exp}\depth{k}(T_1))\, \dots \, (\m{exp}\depth{k}(T_n))$ \\
  $\m{exp}\depth{k+1}(\lambda \bar x.\,H\, \bar y)  = \lambda \bar x.\,H\hrec\bar y$ \\
  $\m{exp}^{\Delta_c}\depth{k+1}(\lambda \bar x.\,\recvar r\, \bar y)  = \m{exp}^{\Delta_c}\depth{k+1}(\lambda \bar x.\,  [\bar y/ \bar z] (h\, \overline{T}))$ if $\recvar r =_d \lambda \bar z.\,h\, \overline{T} \in \Delta_c$\\
  \end{tabular}
\end{center}

As an example, we use an encoding of stream processors $\mt{sp}$ \cite{Ghani09lmcs,Danielsson10mpc,Abel16jfp}. 
At each step, a stream processor 
may choose to consume an input element (\verb$get$) or produce an output element (\verb$put$) and may do so indefinitely.
\footnote{Stream processors were used to illustrate the semantics of mixed-induction and coinduction, 
but here we consider stream processors to be purely coinductive. Thus, we are happy to accept stream 
processors that keep consuming inputs without producing an output. }
The use of $\lambda$-bindings due to the typing of \verb$get$ ensures that 
a stream can only produce elements that it has consumed.
\begin{verbatim}
sp : cotype.
element : type.
get: (element -> sp) -> sp.
put: element -> sp -> sp.
\end{verbatim}
We may define stream processors $\mt{odd}$ or $\mt{even}$ that return only the odd-indexed or even-indexed elements, where the index starts from 0.
We write $\lambda$-abstractions in square brackets, following the convention of CoLF \cite{Chen23fossacs}.
\begin{verbatim}
odd : sp = get ([x] even).
even : sp = get ([x] put x odd).
\end{verbatim}
We may use unification to determine what is the behavior of the stream processor $\mt{S}$
after reading two elements of the input, if it behaves the same as $\mt{odd}$.
\begin{verbatim}
?- get ([x] get ([y] S x y)) = odd.
\end{verbatim}
The problem may be posed as the following concrete unification context $\Delta_c$, which will be used as a 
running example.
\[\Delta_c = \{\mi{get}\, (\lambda x.\, \mi{get}\,(\lambda y. \, S\, x\, y)) \doteq \recvar{odd},
 \recvar{odd}=_d \mi{get}\, (\lambda x.\, \recvar{even}),
 \recvar{even}=_d \mi{get}\, (\lambda x.\, \mi{put}\, x\, \recvar{odd})
 \}
 \]
 Eventually, we will find the following most general unifier, written in the concrete syntax.
\[\Gamma_c  =\{
 S\,x\,y \doteq \recvar{r_3}\, y,
  \recvar{odd}=_d \mi{get}\, (\lambda x.\,\mi{get}\, (\lambda y.\, \recvar{r_3}\, y)), 
  \recvar{r_3} =_d \lambda w.\, \mi{put}\, w\, \recvar{odd},
\}\]

\subsection{Preprocessing}
\label{sec:preprocessing}
As with the first-order case,
we preprocess the unification problem $\Delta_c$ such that every recursive definition will only 
be one level deep. In the higher-order case, this processing is similar to Huet's treatment of regular B\"ohm trees  \citeyearpar{Huet98mscs}. 
Terms are divided into contractive terms $U$, which have constructors ($c, d, e$), bound variables ($x, y, z$), or contractive unification metavariables ($E\hcontra, F\hcontra, G\hcontra, H\hcontra$) as their heads,
and recursive terms $N$, which have either recursion constants ($\recvar{r},\recvar{s},\recvar{t}$) or recursive unification metavariables ($E\hrec, F\hrec, G\hrec, H\hrec$) as their heads.
It is still the case that the terms are always written in $\beta$-normal-$\eta$-long forms, with the exception that 
arguments to recursion constants and unification metavariables are written in non-expanded forms.
The grammar for terms in their preprocessed form is summarized as follows:

\begin{center}
\begin{tabular}{lll}
  Unification Contexts & $\Delta ::= []\mid \Delta, U_1 \doteq U_2 \mid \Delta, N_1 \doteq N_2 \mid \Delta, \recvar r =_d U \mid \Delta, \m{contra}$\\
  Contractive  Terms & $ U ::= \lambda \bar x.\,y\, \overline N \mid \lambda \bar x.\,c \, \overline{N} \mid \lambda \bar x.\,H\hcontra \, \bar y$ \\
  Recursive Terms & $N ::= \lambda \bar x. \, \recvar r\, \bar y \mid \lambda \bar x.\, H\hrec\, \bar y$ \\
\end{tabular}
\end{center}

We use the letter $h$ to denote either constructors $c$ or variables $x$, but not unification metavariables.
We use $FV(U)$ or $FV(N)$ to denote the set of free variables in $U$ or $N$, and  
may use the syntax category $M$ to denote either $U$ or $N$. 
We also require all recursive definitions to be closed in the sense that $\recvar r =_d U \in \Delta$ implies $FV(U) = \emptyset$.
As with the first-order case,  $UV(\Delta)$ denotes 
the set of unification metavariables in $\Delta$. $\m{defs}(\Delta)$ and $\m{eqs}(\Delta)$ denote the list of recursive definitions and equations of $\Delta$ respectively.
$\Delta_1,\Delta_2$ denotes the union of two contexts $\Delta_1$ and $\Delta_2$, consistently renaming recursion constants in $\Delta_2$ if necessary.
$\Delta$ is \emph{contradiction-free} if $\m{contra}\notin\Delta$.

As with first-order terms, we use the judgments
$\Delta_c \rhd \Delta$,
$T \rhd\hcontra U \diamond \Delta$,
$T \rhd\hrec N \diamond \Delta$ for translating from concrete syntax into unification contexts, contractive terms, 
and recursive terms. They are defined as follows.

\resetRN
\boxed{\Delta_c \rhd \Delta}
\begin{mathpar}
  \inferrule{ }{[] \rhd []}\RN

\inferrule{
\Delta_c \rhd \Delta_1 \\
 T_1 \rhd\hrec N_1 \diamond \Delta_2 \\
  T_2 \rhd\hrec N_2 \diamond \Delta_3}
  {\Delta_c, T_1 \doteq T_2 \rhd \Delta_1,  \Delta_2, \Delta_3, N_1 \doteq N_2}
  \RN \label{rule:_xxx_fefef}

  \inferrule{
\Delta_c \rhd \Delta_1 \\  \lambda \bar x.\, h\, \overline{T} \rhd\hcontra U \diamond \Delta_2
  }{
\Delta_c, \recvar r =_d  \lambda \bar x.\, h\, \overline T \rhd \Delta_1, \Delta_2, \recvar r =_d U
  }\RN
\end{mathpar}

\boxed{T \rhd\hrec N \diamond \Delta}
\begin{mathpar}
  \inferrule{
h\, \overline{T} \rhd\hcontra U \diamond \Delta 
\\ \bar z = FV(h\, \overline{T})
  }{
\lambda \bar x.\, h\, \overline{T} \rhd\hrec \lambda \bar x.\, \recvar r\, \bar z \diamond (\Delta, \recvar r =_d\lambda \bar z.\, U)
  }(\recvar r \text{ fresh})\RN
  \label{rule:ho_rec_exp_head_const}
  \label{rule:ho_trans_into_rec}

  \inferrule{
  }{
    \lambda \bar x.\, H\,\bar y \rhd 
    \lambda \bar x.\, H\hrec\,\bar y 
     \diamond []
  }\RN

  \inferrule{
  }{
    \lambda \bar x.\, \recvar{r}\,\bar y \rhd 
    \lambda \bar x.\, \recvar{r}\,\bar y 
     \diamond []
  }\RN
\end{mathpar}

\boxed{T \rhd\hcontra U \diamond \Delta}
\begin{mathpar}
  \inferrule{
    \forall_{i, 1 \le i \le n}. T_i \rhd\hrec N_i \diamond \Delta_i
  }{
\lambda \bar x.\, h\, T_1\, \dots\, T_n \rhd\hcontra 
\lambda \bar x.\, h\, N_1\, \dots\, N_n \diamond (\Delta_1, \dots, \Delta_n)
  }\RN

  (\text{No rules for $T = \lambda \bar x.\, H\,\bar y$ or $T = \lambda \bar x.\, \recvar r\, \bar y$})
\end{mathpar}

In Rule~\ref{rule:ho_trans_into_rec}, we ensure that the body of a recursive definition is always closed by 
abstracting over all free variables when creating a recursive definition.

As an example, we show how the unification problem $\Delta_c$ in the previous section is translated into $\Delta$.
Notice that the left-hand-side of the unification equation $\mi{get}\, (\lambda x.\, \mi{get}\, (\lambda y. \, \mi{S}\, x\, y))$ 
is moved into a recursive definition $\recvar{r_1}$ according to the definition, and the body of $\mi{even}$ is extracted to $\recvar{r_3}$.

\begin{multicols}{2}
\begin{tabular}{ll}
$\Delta_c = \{ $\\
$\ \mi{get}\, (\lambda x.\, \mi{get}\,(\lambda y. \, S\, x\, y)) \doteq \recvar{odd}, $\\
\\
\\
$\ \recvar{odd}=_d \mi{get}\, (\lambda x.\, \recvar{even}), $\\
$\ \recvar{even}=_d \mi{get}\, (\lambda x.\, \mi{put}\, x\, \recvar{odd}) $\\
$\}$\\
\end{tabular}
 
\columnbreak
\begin{tabular}{ll}
$\Delta = \{ $\\
$\ \recvar{r_1} \doteq \recvar{odd}, $\\
$\ \recvar{r_1} =_d \mi{get}\, (\lambda  x.\, \recvar{r_2}\,  x), $\\
$\ \recvar{r_2} =_d \lambda x.\, \mi{get}\, (\lambda  y.\, S\hrec\, x\,y), $\\
$\ \recvar{odd}=_d \mi{get}\, (\lambda x.\, \recvar{even}), $\\
$\ \recvar{even}=_d \mi{get}\, (\lambda x.\, \recvar{r_3}\, x), $\\
$\ \recvar{r_3} =_d \lambda x.\, \mi{put}\, (\recvar{r_4}\, x)\, \recvar{odd},$\\
$\ \recvar{r_4} =_d \lambda x.\, x$\\
$\}$\\
\end{tabular}
\end{multicols}

As with the first-order case, we define the definitional expansion at depth $k$ for a recursive or a contractive term mutually recursively,
 $\m{exp}^{\Delta}\depth k(U) = M_\bot$ and $\m{exp}^{\Delta}\depth k(N) = M_\bot$. We also take the liberty 
 to omit writing $\Delta$ if it 
 remains unchanged throughout and is not referenced.

 \begin{center}
 \begin{tabular}{ll}
 $\m{exp}\depth 0(U) = \bot$ \\
 $\m{exp}\depth{k+1}(\lambda \bar x.\,H\hcontra\, \bar y)  = \lambda \bar x.\,H\hcontra\bar y$ \\
 $\m{exp}\depth{k+1}(\lambda \bar x.\,h\, T_1 \, \dots\, T_n) =\lambda \bar x.\, h\, (\m{exp}\depth{k}(T_1))\, \dots \, (\m{exp}\depth{k}(T_n))$ \\[1ex]
 $\m{exp}\depth 0(N) = \bot$ \\
 $\m{exp}\depth{k+1}(\lambda \bar x.\,H\hrec\, \bar y)  = \lambda \bar x.\,H\hrec\bar y$ \\
 $\m{exp}^{\Delta}\depth{k+1}(\lambda \bar x.\,\recvar r\, \bar y)  = \m{exp}^{\Delta}\depth{k+1}(\lambda \bar x.\,  [\bar y/ \bar z] (h\, \overline{T}))$
 if $\recvar{r} =_d \lambda \bar z.\,h\, \overline{T} \in \Delta$\\
 \end{tabular}
\end{center}

We show that 
the translation preserves the definitional expansion of arbitrary depth.

\begin{theorem}
  We have
  \begin{enumerate}
    \item 
If $T \rhd \hcontra U \diamond \Delta_2$
  and $\m{exp}^{\Delta_c}\depth k (\lambda \bar x.\, \recvar{s}\, \bar y) = \m{exp}^{\Delta_1}\depth k (\lambda \bar x.\, \recvar s \, \bar y)$ for all $\recvar s$ occurring in $T$,  
  then $\m{exp}^{\Delta_c}\depth k (T) = \m{exp}^{\Delta_1, \Delta_2}\depth k(U)$.
    \item 
If $T \rhd \hrec N \diamond \Delta_2$
  and $\m{exp}^{\Delta_c}\depth k (\lambda \bar x.\, \recvar s\, \bar y) = \m{exp}^{\Delta_1}\depth k (\lambda \bar x.\, \recvar s\, \bar y)$ for all $\recvar s$ occurring in $T$,  
  then $\m{exp}^{\Delta_c}\depth k (T) = \m{exp}^{\Delta_1, \Delta_2}\depth k(N)$.
  \end{enumerate}
\end{theorem}

\begin{proof}
  Simultaneous induction on $(k, T)$, where (2) may appeal to (1) without a decrease in size.

  We show the case for rule (\ref*{rule:ho_rec_exp_head_const}) as an example. 
  The premise $h\, \overline{T} \rhd\hcontra U \diamond \Delta$ implies, by the induction hypothesis, that
  $\m{exp}^{\Delta_c}\depth k (h\, \overline T) = \m{exp}^{\Delta_1, \Delta}\depth k(U)$. 
  Let $\bar z = FV(h\, \overline T)$, 
  we want to show that $\m{exp}^{\Delta_c}\depth k (\lambda \bar x.\, h\, \overline T) = \m{exp}^{\Delta_1, \Delta, \recvar r=_d\lambda \bar z.\, U}\depth k(\lambda \bar x.\, \recvar r\, \bar z)$.
  Observe that both $\m{exp}^{\Delta_c}\depth k$ and $\m{exp}^{\Delta}\depth k$ commute with $\lambda$-abstractions, 
  and that $\m{exp}^\Delta\depth k$ is fixed by 
  the recursive definitions for recursion constants that occur in the argument, i.e. $\m{exp}^{\Delta_1, \Delta}\depth k (U) = \m{exp}^{\Delta_1, \Delta, \recvar{r}=_d\lambda \bar z.\, U}\depth k(U)$.
  We have

  \begin{tabular}{ll}
$\m{exp}^{\Delta_c}\depth k (\lambda \bar x.\, h\, \overline T) $ & \\
$= \lambda \bar x.\,\m{exp}^{\Delta_c}\depth k ( h\, \overline T) $ &(by definition) \\
$= \lambda \bar x.\,\m{exp}^{\Delta_1, \Delta}\depth k ( U) $ & (by the induction hypothesis)\\
$= \lambda \bar x.\,\m{exp}^{\Delta_1, \Delta, \recvar r=_d\lambda \bar z.U}\depth k ( U) $ & (because $\recvar r$ is fresh) \\
$=  \lambda \bar x.\,\m{exp}^{\Delta_1, \Delta, \recvar r=_d\lambda \bar z.\, U}\depth k( \recvar r\, \bar z) $ & (by definition)\\
$=  \m{exp}^{\Delta_1, \Delta, \recvar r=_d\lambda \bar z.\, U}\depth k(\lambda \bar x.\, \recvar r\, \bar z)$ & (by definition)\\
  \end{tabular}

\end{proof}

\begin{corollary}
  If $\Delta_c \rhd \Delta$, then every equation in $\Delta_c$ corresponds to an equation in $\Delta$ with an equal 
  definitional denotation, and every recursive definition in $\Delta_c$ corresponds to a recursive definition in $\Delta$.
\end{corollary}

\begin{proof}
  Directly by structural induction over $\Delta_c \rhd \Delta$.
\end{proof}

\subsection{Term Equality and Unifiers}
\label{sec:term_equality_and_unifiers}

The core ideas for term equality, substitution, and unifiers for the higher-order case 
are similar to the first-order case.
The main technical difference from the first-order case is that $H^m\, \bar x \doteq U$ is to be interpreted as $H^m$ standing for $\lambda \bar x.\, U$.
We will only repeat the most essential definitions.

Informally two terms are equal if they have the same definitional expansion.
Formally, $M$ is \emph{equal} to $M'$ in a context $\Delta$ (i.e. $M \doteq M'$ holds) if 
for all $k$, $\m{exp}^{\Delta}\depth k (M) = \m{exp}^{\Delta}\depth k (M')$.

There are two kinds of substitutions in the higher-order case, substitutions for ordinary variables
and substitutions for unification metavariables. Due to the pattern restriction, the only substitutions 
for ordinary variables are simultaneous variable renaming that we have seen, and are written in the notation $[\bar y /\bar x]M$.
Substitutions for unification metavariables remain a special form of unification contexts, but they are now higher-order.

A \emph{substitution}  is a contradiction-free unification context
where the left-hand side of each unification equation is a unique unification metavariable followed by a list of bound variables, 
which is a superset of the free variables occurring on the right-hand side of that equation. 
Intuitively, the variables following a unification metavariable serve as $\lambda$-binders for its value (on the right-hand side).
All unification equations in a substitution are of the form $H\hcontra\, \bar x \doteq U$ or $H\hrec\, \bar x \doteq N$, 
where $FV(U) \subseteq \bar x$ and $FV(N) \subseteq \bar x$. The equation $H\hcontra\, \bar x \doteq U$ or $H\hrec\, \bar x \doteq N$ 
is called a \emph{substitution equation} for $H\hrec$ or $H\hcontra$ in $\Gamma$. The intuitive meaning of a substitution equation $H^m\, \bar x\doteq M$ is that $H^m$ ``stands for''  $\lambda \bar x.\, M$.
Since terms are all written in their $\eta$-long-form, $U$ or $N$ on the right-hand side of the unification equation should not contain top-level $\lambda$-bindings.
The set of unification metavariables that occur on the left-hand sides of the unification equations in a substitution $\Gamma$ is 
called the \emph{domain} of the substitution and is denoted $\m{dom}(\Gamma)$.
Two substitutions $\Gamma$ and $\Gamma'$ are equal 
if $\m{dom}(\Gamma) = \m{dom}(\Gamma')$, 
and for all $H^m \in \m{dom}(\Gamma)$, $\m{exp}^\Gamma\depth{k}((\lambda \bar x.\, H^m\, \bar x)[\Gamma]) = \m{exp}^{\Gamma'}\depth{k}((\lambda \bar x.\, H^m \bar x)[\Gamma])$, for all $k$, 
where $\lambda \bar x.\, H^m \, \bar x$ is the $\eta$-long-form of $H^m$ according to its simple type.

As an example, the $\Gamma$, $\Gamma'$ and $\Gamma''$ below are all equal substitutions with the domain $\{S\hrec\}$. The main differences are highlighted in \blue{blue}.

\begin{multicols}{3}
\begin{tabular}{ll}
$\Gamma =\{$ \\
$\ \blue{S\hrec\,z\,w \doteq \mi{r_3}\, w}, $\\
$\ \recvar{odd}=_d \mi{get}\, (\lambda x.\, \recvar{even}), $\\
$\ \recvar{even}=_d \mi{get}\, (\lambda x.\, \recvar{r_3}\, x), $\\
$\ \recvar{r_3} =_d \lambda x.\, \mi{put}\, (\recvar{r_4}\, x)\, \recvar{odd}, $\\
$\ \recvar{r_4} =_d \lambda x.\, x $\\
$\}$\\
\end{tabular}

\columnbreak
\begin{tabular}{ll}
$\Gamma' = \{ $\\
$\ \blue{S\hrec\,y\,u \doteq \recvar{r_3}\, u}, $\\
$\ \recvar{odd}=_d \mi{get}\, (\lambda x.\, \recvar{even}), $\\
$\ \recvar{even}=_d \mi{get}\, (\lambda x.\, \recvar{r_3}\, x),$\\
$\ \recvar{r_3} =_d \lambda x.\, \mi{put}\, (\recvar{r_4}\, x)\, \recvar{odd}, $\\
$\ \recvar{r_4} =_d \lambda x.\, x$\\
$\}$\\
\end{tabular}

\columnbreak
\begin{tabular}{ll}
$\Gamma'' = \{ $\\
$\ \blue{S\hrec\,w\,z \doteq \recvar{r_1}\, z}, $\\
$\ \recvar{odd}=_d \mi{get}\, (\lambda x.\, \recvar{even}), $\\
$\ \recvar{even}=_d \mi{get}\, (\lambda x.\, \blue{\recvar{r_1}}\, x),$\\
$\ \blue{\recvar{r_1}} =_d \lambda x.\, \mi{put}\, (\blue{\recvar{r_5}}\, x)\, \recvar{odd}, $\\
$\ \blue{\recvar{r_5}} =_d \lambda x.\, x $\\
$\}$
\end{tabular}
\end{multicols}

We write $U[\Gamma]$ and $N[\Gamma]$ for applying the substitution (for unification metavariables) to terms. They are defined as follows.

\begin{tabular}{ll}
 $(\lambda \bar x.h\,  N_1\, \dots\,N_n)[\Gamma] = \lambda \bar x.\, h\, (N_1[\Gamma])\, \dots\, (N_n[\Gamma])$ \\
 $(\lambda \bar x.\, H\hcontra\, \bar y)[\Gamma] = \begin{cases}
   \lambda \bar x.\, [\bar y/\bar z]U' .           & \text{ if } H\hcontra\,\bar z\doteq U' \in \m{eqs}(\Gamma)\\
   \lambda \bar x.\, H\hcontra\, \bar y  & \text{ otherwise} \\
 \end{cases}$ \\[1ex]
 $(\lambda \bar x.\,\recvar r\, \bar y )[\Gamma] = \lambda \bar x.\,\recvar r\, \bar y $ \\
 $(\lambda \bar x.\, H\hrec\, \bar y)[\Gamma] = \begin{cases}
   \lambda \bar x.\, [\bar y/\bar z]N'    & \text{ if } H\hrec\, \bar z \doteq N' \in \m{eqs}(\Gamma)\\
   \lambda \bar x.\, H\hrec\, \bar y  & \text{ otherwise} \\
 \end{cases}$ \\
\end{tabular}

It is worth noting that the substitution commutes with $\lambda$-abstractions so that $(\lambda \bar x.\, M)[\Gamma] = \lambda \bar x.\, (M[\Gamma])$.
Substitution also commutes with simultaneous variable renaming so that $([\bar y/\bar x]M)[\Gamma] = [\bar y/\bar x](M[\Gamma])$. 
We can show both claims by induction on the structure of $M$, and the intuition is that substitutions are ``closed'' substitutions for unification metavariables.

The application of a substitution $\Gamma$ to a unification context $\Delta$, 
the application of a substitution $\Gamma_2$ to another substitution $\Gamma_1$, 
and the composition of substitutions $\Gamma_1 \circ \Gamma_2$ (apply $\Gamma_1$ and then $\Gamma_2$) 
are analogous counterparts of their first-order definitions.

 \begin{tabular}{ll}
$\Delta[\Gamma] $ & $= \m{defs}(\Gamma), \{M[\Gamma] \doteq M'[\Gamma] \mid M \doteq M' \in \m{eqs}(\Delta)\},  \{\recvar r =_d U[\Gamma] \mid \recvar r =_d U \in \m{defs}(\Delta)\}$ \\
$\Gamma_1[\Gamma_2]$ & $ = \m{defs}(\Gamma_2),  \{H^m\, \bar x \doteq M'[\Gamma_2] \mid H^m\, \bar x \doteq M' \in \m{eqs}(\Gamma_1)\},  \{\recvar r =_d U[\Gamma_2] \mid \recvar r =_d U \in \m{defs}(\Gamma_1)\}$ \\
$ \Gamma_1 \circ \Gamma_2$ &$  = (\Gamma_1[\Gamma_2]), \{H^m\, \bar x \doteq M \mid H^m\, \bar x \doteq M \in \m{eqs}(\Gamma_2) \land H^m \notin \m{dom}(\Gamma_1)\} $ \\
 \end{tabular}

We repeat the definition of unifiers and the most general unifier for the higher-order case.
A \emph{unifier} for a contradiction-free unification context $\Delta$ is a substitution $\Gamma$ such that $UV(\Delta) = \m{dom}(\Gamma)$, and every equation in $\Delta[\Gamma]$ holds. 
A unification context $\Delta$ with $\m{contra} \in \Delta$ has no unifiers.
A unifier $\Gamma_1$ is \emph{more general} than another unifier $\Gamma_2$  if
there is a substitution $\Gamma'$ such that $\Gamma_1 \circ \Gamma' = \Gamma_2$.

As an example, given $\Gamma$ and $\Delta$ defined below, $\Gamma$ is a unifier of $\Delta$, because every 
equation in $\Delta[\Gamma]$ holds. The main changes are highlighted in \blue{blue}. Notice that when carrying out the substitution, the duplicate recursion constants
in $\Gamma$ are renamed by adding a prime ($'$) sign.

\newpage 
\begin{multicols}{3}
\begin{tabular}{ll}
$\Gamma  =\{$ \\
$\ \blue{S\hrec\,z\,w} \doteq \recvar{r_3}\, w,$ \\
\\
\\
\\
\\
\\
\\
$\ \recvar{odd}=_d \mi{get}\, (\lambda x.\, \recvar{even}),$ \\
$\ \recvar{even}=_d \mi{get}\, (\lambda x.\, \recvar{r_3}\, x),$ \\
$\ \recvar{r_3} =_d \lambda x.\, \mi{put}\, (\recvar{r_4}\, x)\, \recvar{odd},$ \\
$\ \recvar{r_4} =_d \lambda x.\, x$ \\
$\}$\\
\end{tabular}

\columnbreak
\begin{tabular}{ll}
$\Delta = \{$ \\
$\ \recvar{r_1} \doteq \recvar{odd}, $ \\
$\ \recvar{r_1} =_d \mi{get}\, (\lambda  x.\, \recvar{r_2}\,  x), $ \\
$\ \recvar{r_2} =_d \lambda x.\, \mi{get}\, (\lambda  y.\, \blue{S\hrec\, x\,y}), $\\
$\ \recvar{odd}=_d \mi{get}\, (\lambda x.\, \recvar{even}),$ \\
$\ \recvar{even}=_d \mi{get}\, (\lambda x.\, \recvar{r_3}\, x),$ \\
$\ \recvar{r_3} =_d \lambda x.\, \mi{put}\, (\recvar{r_4}\, x)\, \recvar{odd},$ \\
$\ \recvar{r_4} =_d \lambda x.\, x$ \\
$\}$\\
\end{tabular}

\columnbreak
\begin{tabular}{ll}
$\Delta[\Gamma] = \{$ \\
$\ \recvar{r_1} \doteq \recvar{odd}, $ \\
$\ \recvar{r_1} =_d \mi{get}\, (\lambda  x.\, \recvar{r_2}\,  x), $ \\
$\ \recvar{r_2} =_d \lambda x.\, \mi{get}\, (\lambda  y.\, \blue{\recvar{r_3}'\,y}), $\\
$\ \recvar{odd}=_d \mi{get}\, (\lambda x.\, \recvar{even}), $\\
$\ \recvar{even}=_d \mi{get}\, (\lambda x.\, \recvar{r_3}\, x), $\\
$\ \recvar{r_3} =_d \lambda x.\, \mi{put}\, (\recvar{r_4}\, x)\, \recvar{odd}, $\\
$\ \recvar{r_4} =_d \lambda x.\, x$ \\
$\ \blue{\recvar{odd}'}=_d \mi{get}\, (\lambda x.\, \blue{\recvar{even}'}), $\\
$\ \blue{\recvar{even}'}=_d \mi{get}\, (\lambda x.\, \blue{\recvar{r_3}'}\, x), $\\
$\ \blue{\recvar{r_3}'} =_d \lambda x.\, \mi{put}\, (\blue{\recvar{r_4}'}\, x)\, \blue{\recvar{odd}'}, $\\
$\ \blue{\recvar{r_4}'} =_d \lambda x.\, x $\\
$\}$\\
\end{tabular}
\end{multicols}

\subsection{The Algorithm}
\label{sec:algorithm}

We now present the saturation-based algorithm for higher-order rational terms.

In the first-order case, a unification metavariable is resolved if there is an equation between the unification metavariable and 
either a recursion constant or a term with a constructor head. In the higher-order case, we have to consider 
the binding structure, the resolution must only have free variables that appear in the arguments to the unification metavariable. 
Also, a unification metavariable may be resolved by another unification metavariable that has strictly fewer arguments.
A contractive unification metavariable $H\hcontra$ is \emph{resolved} if (a) there exists a unification equation $H\hcontra\, \bar y \doteq h\, \bar z$ (or its symmetry) and $\bar z \subseteq \bar y$, 
or (b) there exists a unification equation $H\hcontra \, \bar y \doteq G\hcontra\, \bar w$ (or its symmetry)
with $\bar w \subsetneq \bar y$.
\footnote{$\bar w\subsetneq\bar y$ means $\bar w$ is a proper subset of $\bar y$. 
For example, $H\hcontra \,y \doteq G\hcontra \,y$ is not a resolution equation, but $H\hcontra y \doteq G\hcontra$ is, while either equation may appear as a substitution equation for $H\hcontra$ in a unifier.
 } 
 $H\hcontra$ is unresolved otherwise, and is denoted by the judgment $H\hcontra \m{unresolved}$. 
Similarly, a contractive unification metavariable $H\hrec$ is \emph{resolved} if there exists an equation $H\hrec \, \bar y \doteq \recvar r\, \bar z$ (or its symmetry) and $\bar z \subseteq \bar y$, 
or there exists an equation $H\hrec \, \bar y \doteq G\hrec\, \bar w$ (or its symmetry) with $\bar w \subsetneq \bar y$.
 $H\hrec$ is unresolved otherwise, and is denoted by the judgment $H\hrec \m{unresolved}$.

We saturate the unification context $\Delta$ using the rules defined in \figurename~\ref{fig:unification_rules}. The saturation rules preserve the definition of all recursion constants.
Once saturated, a unifier can be constructed easily.
The presence of the symbol $\m{contra}$ in a 
unification context indicates that the unification context has no unifiers.

We use the concept of a parameter to ensure the termination of the saturation-based algorithm. The parameters are indicated by bracketed existential 
quantifiers $(\exists X)$ where $X$ is a parameter that stands for 
a variable, a recursion constant, a unification metavariable, or a list of those.
The new equations or definitions under the existential quantification subsume
any instantiation of the equation or definition \cite{McLaughlin09}.
The parameters ensure freshness and non-redundancy:
when new variables, recursion constants, or unification metavariables are introduced by one of the rules (freshness), 
the existential quantification 
ensures that the rule applies (thus the conclusion equations or definitions are created) only if there does not exist any instantiation of 
the conclusion equations or definitions in the unification context (non-redundancy).
For example, $(\exists x, y) \, c\, x \doteq d\, y$ means that we pick globally fresh variables $x, y$, ensure that 
there is no equation of the form $c\, z \doteq d\, w$ in the unification context for any variable $z$ and $w$, as this is a renaming of $c\, x= d\, y$, and 
then add the equation $c\, x \doteq d\, y$ to the unification context.

We extend the notation of using overlines to denote lists of unification metavariables (possibly with arguments) and operations on them.
A list of unification metavariables is written $\overline{G\hrec}$ or $\overline{G\hcontra}$.
 Also, we write $\overline{G\hrec\, \overline{x}}$ to denote the list of applications where each unification metavariable is applied to $\overline{x}$.
 For example, $c\, \overline{G\hrec\, \overline{x}}$ denotes $c\, (G\hrec_1\, x_1\, \dots\, x_m)\, \dots (G\hrec_n\, x_1\, \dots\, x_m)$. 
 $h\,\overline{\m{\eta exp}(\overline{G\hrec\, \overline{x}})}$ denotes a term with head $h$ whose arguments are 
 the result of applying top-level $\eta$-expansions to all terms in $\overline{G\hrec\, \overline{x}}$ according to the simple type of $h$.
 For example, $c\, \overline{\m{\eta exp}(\overline{G\hrec\, \overline{x}})}$
 \footnote{We should remark that the non-$\eta$-expanded version $c\,\overline{G\hrec\, \overline{x}}$ should not appear in the unification context, since we write everything in the $\eta$-long form (except the arguments to 
 recursion constants and unification metavariables). Thus, $h\,\overline{\m{\eta exp}(\overline{G\hrec\, \overline{x}})}$ always appears in a conclusion where $\overline{G\hrec}$ is fresh (bound by the existential quantifier $(\exists \overline{G\hrec})$).}
 denotes a term of the form 
 $c\, (\lambda \overline{w}_1.\, G\hrec_1\, \overline{x}\, \overline{w}_1)\, \dots\, (\lambda \overline{w}_n.\, G\hrec_n\, \overline{x}\, \overline{w}_n)$, i.e.,
 \[c\, (\lambda w_{1, 1}.\dots\lambda w_{1,l_1}.\, G\hrec_1\, x_1\,\dots\,x_m\,w_{1,1}\, \dots\, w_{1,l_1})\, \dots\, (\lambda w_{n, 1}.\dots\lambda w_{n,l_n}.\, G\hrec_n\, x_1\,\dots\,x_m\,w_{n,1}\, \dots\, w_{n,l_n})\]

\begin{figure}
\resetRN
\noindent\rule{\textwidth}{1pt}
\textbf{Structural Rules}
\begin{mathpar}
\inferrule[\RNwithLabel{U-INST}\label{rule:ho_inst_contra}]{
  \lambda \bar x.\, U \doteq \lambda \bar x.\, U'
}{
    (\exists \bar{x})\,U \doteq  U'
}

\inferrule[\RNwithLabel{N-INST}\label{rule:ho_inst_rec}]{
  \lambda \bar x.\, N \doteq \lambda \bar x.\, N'
}{
  (\exists \bar{x})\, N \doteq  N'
}

\inferrule[\RNwithLabel{SIMP-F1}\label{rule:ho_structural_var_const}]{
  x\, \overline N \doteq c\, \overline{N'}
}{\m{contra}}

\inferrule[\RNwithLabel{SIMP-F2}\label{rule:ho_structural_var_var_clash}]{
  x\, \overline N \doteq y\, \overline{N'}
}{\m{contra}}(x \ne y)

\inferrule[\RNwithLabel{SIMP-F3}\label{rule:ho_structural_const_const_clash}]{
  c\, \overline N \doteq d\, \overline{N'}
}{\m{contra}}(c \ne d)

\inferrule[\RNwithLabel{SIMP}\label{rule:ho_structural}]{
  h\, N_1\, \dots N_n \doteq
  h\, N_1'\, \dots N_n'
}{
  N_1 \doteq N_1', \dots, N_n \doteq N_n'
}

  \inferrule[\RNwithLabel{PROJ-F}\label{rule:ho_projection_fail}]{
    H\hcontra\, \bar y \doteq z_i\, \overline{N}\\ z_i \notin \bar y
  }{
    \m{contra}
  }
\end{mathpar}

\noindent\rule{\textwidth}{1pt}
\textbf{Resolution Rules}
\begin{mathpar}
  \inferrule[\RNwithLabel{IMIT}\label{rule:ho_immitation}]{
    H\hcontra\, \bar y \doteq c\, \overline{N}
    \\ H\hcontra \, \m{unresolved}
  }{
    (\exists \overline{G\hrec})\,
    H\hcontra\, \bar y \doteq 
c\, \overline{\m{\eta exp}(\overline{G\hrec\, \bar{y}})}\, 
  }

  \inferrule[\RNwithLabel{PROJ}\label{rule:ho_projection}]{
    H\hcontra\, \bar y \doteq y_i\, \overline{N}\\ y_i \in \bar y
    \\ H\hcontra \, \m{unresolved}
  }{
    (\exists \overline{G\hrec}) \, 
    H\hcontra\, \bar y \doteq 
    y_i\, \overline{\m{\eta exp}(\overline{G\hrec\, \bar{y}})}\, 
  } 

  \inferrule[\RNwithLabel{PRUNE}\label{rule:ho_rec_pruning}]{
H\hrec\, \bar y \doteq  \recvar r\, \bar x 
     \\  \bar x \not\subseteq \bar y 
    \\ \bar w = \bar x \cap \bar y
    \\ H\hrec \m{unresolved}
  }{
    (\exists t)(\exists G\hcontra) \, H\hrec\, \bar y \doteq \recvar t\,  \bar w,
    \recvar{r}\, \bar x \doteq \recvar t\, \bar w, \recvar t=_d \lambda \bar w. G\hcontra\,  \bar w
  }

  \inferrule[\RNwithLabel{FF-D}\label{rule:ho_flexible_flexible_different_vars}]{
    G^m \, \bar x \doteq H^m \, \bar y
    \\ G^m \ne H^m
    \\ \bar z = \bar x \cap \bar y 
    \\ \bar x \not\subseteq \bar y  \land \bar y \not\subseteq \bar x
    \\ H^m \m{unresolved} \lor G^m\m{unresolved}
  }{
    (\exists F^m)\, G^m \, \bar x \doteq F^m \, \bar z,
    H^m \, \bar y \doteq F^m \, \bar z
  }(m~\in~\{\symbcontra, \symbrec\})

  \inferrule[\RNwithLabel{FF-S}\label{rule:ho_flexible_flexible_same_var}]{
    G^m \, \bar x \doteq G^m \, \bar y 
    \\ \bar x = x_1\, \dots\, x_n 
    \\ \bar y = y_1\, \dots\, y_n 
    \\ \bar z = \cup_i \{x_i \mid x_i = y_i\} 
    \\ \bar x \neq \bar y 
    \\ G^m \m{unresolved}
  }{
    (\exists F^m)\, G^m \, \bar x \doteq F^m \, \bar z,
    G^m \, \bar y \doteq F^m \, \bar z
  }(m~\in~\{\symbcontra, \symbrec\})
\end{mathpar}

\noindent\rule{\textwidth}{1pt}
\textbf{Expansion, Consistency, Symmetry, and Transitivity}%
\begin{mathpar}
  \inferrule[\RNwithLabel{REC-EXP}\label{rule:ho_rec_expansion}]{
    \recvar r\, \bar x\doteq \recvar s\, \bar y\\
    \recvar r =_d \lambda \bar z.\, U_1 \\
    \recvar s =_d \lambda \bar w.\, U_2 \\
  }{
    [\bar x/\bar z] U_1\doteq [\bar y / \bar w]U_2
  }

  \inferrule[\RNwithLabel{U-AGREE}\label{rule:ho_unif_consistency_contra}]{
    H\hcontra\,\bar x \doteq  U_1 \\
    H\hcontra\,\bar y\doteq   U_2 
    \\  FV(U_1) \subseteq  \bar x
    \\  FV(U_2) \subseteq  \bar y
  }{
     U_1\doteq [\bar x / \bar y] U_2
  }

  \inferrule[\RNwithLabel{N-AGREE}\label{rule:ho_unif_consistency_rec}]{
    H\hrec\,\bar x \doteq  N_1 \\
    H\hrec\,\bar y\doteq   N_2
    \\  FV(N_1) \subseteq \bar x 
    \\  FV(N_2) \subseteq \bar y 
  }{
     N_1\doteq [\bar x / \bar y] N_2
  }

\inferrule[\RNwithLabel{U-SYM}\label{rule:ho_symmetry_contra}]{
  U \doteq U'
}{U' \doteq U}

\inferrule[\RNwithLabel{N-SYM}\label{rule:ho_symmetry_rec}]{
  N \doteq N'
}{N' \doteq N}

\inferrule[\RNwithLabel{U-TRANS}\label{rule:ho_transitivity_contra}]{
  U_1 \doteq U_2 \\ U_2 \doteq U_3 
}{U_1 \doteq U_3}

\inferrule[\RNwithLabel{N-TRANS}\label{rule:ho_transitivity_rec}]{
  N_1 \doteq N_2 \\ N_2 \doteq N_3 
}{N_1 \doteq N_3}
\end{mathpar}
\noindent\rule{\textwidth}{1pt}
\caption{Unification Rules}
\label{fig:unification_rules}
\end{figure}

Rule (\ref*{rule:ho_inst_contra})(\ref*{rule:ho_inst_rec}) instantiate $\lambda$-abstractions, but only when 
there are no instantiations that are currently present in the context.
Rule (\ref*{rule:ho_immitation}) is called imitation by \citet{Huet75tcs} because $H\hcontra$ mimics the behavior 
of the term $c\, \overline{N}$ on the right. Rules (\ref*{rule:ho_projection_fail})(\ref*{rule:ho_projection}) are 
called projections because $H\hcontra$ projects its $i$th argument its head.
Rule (\ref*{rule:ho_rec_pruning})
 prunes the variables that are not in common from both $H\hrec$ and $\recvar r$, by creating 
 a recursive definition $\recvar t$, whose body $G\hcontra\, \bar w$ can only use variables $\bar w$ that 
 are common to $H\hrec$ and $\recvar r$.
Rules (\ref*{rule:ho_flexible_flexible_different_vars})(\ref*{rule:ho_flexible_flexible_same_var})
\footnote{(\ref*{rule:ho_flexible_flexible_different_vars}) suggests flexible-flexible pairs with different unification metavariables and (\ref*{rule:ho_flexible_flexible_same_var}) suggests flexible-flexible pairs with the same unification metavariable. See \cite{Huet75tcs} for the distinction between flexible and rigid terms.}
 also serve a similar effect of pruning by 
removing the extra variables from the arguments to unification metavariables that cannot be used. It resolves $H^m$, or $G^m$, or both, by equating them
to a common unification metavariable $F^m$ whose arguments $\bar z$ is a strict subset of both $\bar x$ and $\bar y$.
Rule (\ref*{rule:ho_rec_expansion}) unfolds recursive definitions and compares them for equality.
Rules (\ref*{rule:ho_unif_consistency_contra})(\ref*{rule:ho_unif_consistency_rec}) ensure that any two resolutions of a unification metavariable are consistent.

We give an example of how the algorithm operates on the stream processor unification context, now denoted $\Delta_7$ (equations and definitions $(1)-(7)$). 
At each step, we show some additional equations or definitions
and the ways they are obtained. We omit the final uninteresting steps when only symmetry and transitivity rules can be applied.

\begin{longtable}{ll}
$(1)\, \recvar{r_1} \doteq \recvar{odd} $ & given \\
$(2)\, \recvar{r_1} =_d \mi{get}\, (\lambda  x.\, \recvar{r_2}\,  x) $ \\
$(3)\, \recvar{r_2} =_d \lambda x.\, \mi{get}\, (\lambda  y.\, S\hrec\, x\,y)$ \\
 $(4)\, \recvar{odd}=_d \mi{get}\, (\lambda x.\, \recvar{even})$ \\
 $(5)\, \recvar{even}=_d \mi{get}\, (\lambda x.\, \recvar{r_3}\, x)$ \\
 $(6)\, \recvar{r_3} =_d \lambda x.\, \mi{put}\, (\recvar{r_4}\, x)\, \recvar{odd}$ \\
 $(7)\, \recvar{r_4} =_d \lambda x.\, x$ \\
 $(8)\, \mi{get}\, (\lambda x.\, \recvar{r_2}\, x) \doteq \mi{get}\, (\lambda x.\, \recvar{even})$ & by Rule (\ref*{rule:ho_rec_expansion}) on $(1)$, $(2)$ and $(4)$ \\
 $(9)\,\lambda x.\, \recvar{r_2}\, x \doteq \lambda x.\, \recvar{even}$ & by Rule (\ref*{rule:ho_structural}) on $(8)$ \\
 $(10)\, \recvar{r_2}\, z \doteq \recvar{even}$ & by Rule (\ref*{rule:ho_inst_rec}) on $(9)$, and we verify that 
  there does not \\ & exist any equation  $(\exists x)\recvar{r_2}\, x \doteq \recvar{even}$ in the context $\Delta_9$. 
\\
 $(11)\, \mi{get}\, (\lambda y.\, S\hrec\, z\, y) \doteq \mi{get}\, (\lambda x.\, \recvar{r_3}\, x)$ & by Rule (\ref*{rule:ho_rec_expansion}) on $(10)$, $(3)$ and $(5)$ \\
 $(12)\,\lambda y.\, S\hrec\, z\, y\doteq\lambda x.\, \recvar{r_3}\, x$ & by Rule (\ref*{rule:ho_structural}) on $(11)$ \\
 $(13) \,S\hrec\,z\,w \doteq \recvar{r_3}\, w$ & by Rule (\ref*{rule:ho_inst_rec}) on $(12)$, and we verify that 
  there does not \\ &exist any equation  $(\exists x)\,S\hrec\,z\,x \doteq \recvar{r_3}\, x$ in the context $\Delta_{12}$. 
\\
$(14)\, \dots$ & by Rules (\ref*{rule:ho_symmetry_contra})(\ref*{rule:ho_symmetry_rec})(\ref*{rule:ho_transitivity_contra})(\ref*{rule:ho_transitivity_rec})$\dots$

\end{longtable}

The above example used only the structural rules and expansion rules. 
More examples will be given in Section~\ref{sec:ho_additional_examples}.

\subsection{Saturated Unification Contexts}
We now describe how a substitution $\Gamma = \m{unif}(\Delta)$ can be constructed from a contradiction-free unification context $\Delta$. 
The construction in the higher-order case takes the pattern variables following a unification metavariable into account.
We will later show that if $\Delta$ is a saturated unification context, then $\m{unif}(\Delta)$ is the most general unifier for $\Delta$.

Given any unification context, every unification metavariable $H\hcontra$ or $H\hrec$ is either resolved or unresolved.
Suppose $H\hcontra$ or $H\hrec$ is resolved.
\begin{enumerate}
  \item 
If $H\hcontra$ is resolved, then there exists an equation $H\hcontra\, \bar y \doteq h\, \bar z$ or $H\hcontra\, \bar y \doteq G\hcontra\, \bar w$, with $ FV(h\, \bar z)\subseteq \bar y$ and $ \bar w \subsetneq \bar y$.
The equation $H\hcontra \, \bar y \doteq h\, \bar z$ or $H\hcontra\, \bar y \doteq G\hcontra\, \bar w$ is called a \emph{resolution equation} for $H\hcontra$. 
\item
If $H\hrec$ is resolved, then there exists an equation $H\hrec\, \bar y \doteq \recvar r\, \bar z$ or $H\hrec\, \bar y \doteq G\hrec\, \bar w$, with $ FV(\recvar r\, \bar z)\subseteq \bar y$ and $ \bar w \subsetneq \bar y$.
The equation $H\hrec \, \bar y \doteq \recvar r\, \bar z$ or $H\hrec \, \bar y \doteq G\hrec\, \bar w$ is called a \emph{resolution equation} for $H\hrec$.
\end{enumerate}
There may be multiple such equations, but we consistently pick a resolution equation (called \emph{the} resolution equation, or simply the \emph{resolution}\footnote{In the higher-order case, 
resolutions for unification metavariables are in the form of equations, whereas in the first-order case, the resolutions are terms.}) for each unification metavariable.
\begin{enumerate}
\item
Further,
the operation of \emph{replacing occurrences of a (resolved) contractive unification metavariable} $H\hcontra$  \emph{by its resolution in a term} $M$ is defined to be 
$M$ with occurrences of $H\hcontra \, \bar x$ replaced by $[\bar x/\bar y](h\, \bar z)$ or $[\bar x/\bar y](G\hcontra\, \bar w)$ according to the resolution equation $H\hcontra \, \bar y \doteq h\, \bar z$ or $H\hcontra\, \bar y \doteq G\hcontra\, \bar w$.
\item
Similarly, 
the operation of \emph{replacing occurrences of a (resolved) recursive unification metavariable} $H\hrec$  \emph{by its resolution in a term} $M$ is defined to be
$M$ with occurrences of $H\hrec \, \bar x$ replaced by $[\bar x/\bar y](\recvar r\, \bar z)$ or $[\bar x/\bar y](G\hrec\, \bar w)$ according to the resolution equation $H\hrec\, \bar y \doteq \recvar r\, \bar z$ or $H\hrec\, \bar y \doteq G\hrec\, \bar w$.
\end{enumerate}

Now suppose $H^m$ ($m \in \{\symbcontra, \symbrec\}$) is unresolved. Unresolved 
unification metavariables form equivalence classes related by $\doteq$. We pick a \emph{representative unification metavariable} 
for each equivalence class, such that if 
$G^m$ is the representative unification metavariable for the equivalence class of $H^m$,  there is an equation $H^m \, \bar y \doteq G^m \, \bar z$, 
 where $\bar z$ is a permutation of $\bar y$. This equation is called the \emph{representative equation} for $H^m$.
 We pick the representative equation in such a way that 
 the right-hand sides $G^m \, \bar z$ of all representative equations are equal for all unresolved unification metavariables in the same equivalence class, i.e.
 for any other unification metavariable $F^m$ in the same equivalence class as $H^m$, $F^m \, \bar w \doteq G^m \, \bar z$ is picked (and $F^m \, \bar w' \doteq G^m \, \bar z'$ with $\bar z' \ne \bar z$ is not picked).
The operation of \emph{replacing occurrences of a (unresolved) unification metavariable} $H^m$ \emph{by its representative unification metavariable} $G^m$ \emph{in a term} $M$ is defined to be 
$M$ with occurrences of $H^m \, \bar x$ replaced by $[\bar x/\bar y](G^m \, \bar z)$.

We construct the substitution $\Gamma = \m{unif}(\Delta)$ for a contradiction-free context $\Delta$ as follows. The construction is 
very similar to the first-order case, except that now the replacement of unification metavariables uses variable renaming.

\begin{enumerate}
  \item Start with $\Gamma$ containing all recursive definitions of $\Delta$.
  \item For each resolved unification metavariable in $UV(\Delta)$, add to $\Gamma$ the resolution equation of that unification metavariable.
  \item For each unresolved unification metavariable in $UV(\Delta)$, add to $\Gamma$ the representative equation of that unification metavariable.
  \item Replace the occurrences of resolved unification metavariables in the right-hand sides and recursive definitions of $\Gamma$ with their resolutions, 
  and replace the occurrences of unresolved unification metavariables in the right-hand sides and recursive definitions of $\Gamma$ with their representatives.
  Repeat this step until all unification metavariables in the right-hand sides and recursive definitions are representative unification metavariables for some 
  equivalence class of unresolved unification metavariables. 
\end{enumerate}

As an example, we show how the unifier for $\Delta_{13}$ (equations and definitions $(1)-(13)$ defined above in Section~\ref{sec:algorithm}), $\Gamma = \m{unif}(\Delta_{13})$ can be constructed.
Major changes in each step are highlighted in \blue{blue}.

\begin{enumerate}
  \item Initialize $\Gamma_1$ to all recursive definitions of $\Delta_{13}$.

  \item Since $H\hrec$ is resolved, we add its resolution equation to get $\Gamma_2$.

\item There are no unresolved unification metavariables, we skip step (3).

\item Replace occurrences of resolved unification metavariables with their resolutions to get $\Gamma_3$.

\begin{minipage}{\textwidth}
\begin{multicols}{3}
\begin{tabular}{ll}
$\Gamma_1 = \{ $\\
$\  \recvar{r_1} =_d \mi{get}\, (\lambda  x.\, \recvar{r_2}\,  x), $ \\
$\  \recvar{r_2} =_d \lambda x.\, \mi{get}\, (\lambda  y.\, S\hrec\, x\,y),$ \\
$\   \recvar{odd}=_d \mi{get}\, (\lambda x.\, \recvar{even}),$ \\
$\  \recvar{even}=_d \mi{get}\, (\lambda x.\, \recvar{r_3}\, x),$ \\
$\  \recvar{r_3} =_d \lambda x.\, \mi{put}\, (\recvar{r_4}\, x)\, \recvar{odd},$ \\
$\  \recvar{r_4} =_d \lambda x.\, x$ \\
$\}$ 
\end{tabular}

\columnbreak
\begin{tabular}{ll}
$\Gamma_2 = \{ $\\
$\   \blue{S\hrec \, z\, w \doteq \recvar{r_3} \, w}, $\\
$\  \recvar{r_1} =_d \mi{get}\, (\lambda  x.\, \recvar{r_2}\,  x),  $\\
$\  \recvar{r_2} =_d \lambda x.\, \mi{get}\, (\lambda  y.\, S\hrec\, x\,y), $\\
$\   \recvar{odd}=_d \mi{get}\, (\lambda x.\, \recvar{even}), $\\
$\   \recvar{even}=_d \mi{get}\, (\lambda x.\, \recvar{r_3}\, x), $\\
$\   \recvar{r_3} =_d \lambda x.\, \mi{put}\, (\recvar{r_4}\, x)\, \recvar{odd}, $\\
$\   \recvar{r_4} =_d \lambda x.\, x $\\
$\}$ 
\end{tabular}

\columnbreak
\begin{tabular}{ll}
$\Gamma_3 = \{ $\\
$\  S\hrec \, z\, w \doteq \recvar{r_3} \, w, $\\
$\  \recvar{r_1} =_d \mi{get}\, (\lambda  x.\, \recvar{r_2}\,  x),  $\\
$\  \recvar{r_2} =_d \lambda x.\, \mi{get}\, (\lambda  y.\, \blue{\recvar{r_3}\,y}), $\\
$\   \recvar{odd}=_d \mi{get}\, (\lambda x.\, \recvar{even}),$\\
$\   \recvar{even}=_d \mi{get}\, (\lambda x.\, \recvar{r_3}\, x), $\\
$\   \recvar{r_3} =_d \lambda x.\, \mi{put}\, (\recvar{r_4}\, x)\, \recvar{odd}, $\\
$\   \recvar{r_4} =_d \lambda x.\, x $\\
$\}$ 
\end{tabular}
\end{multicols}
\end{minipage}

\item Note that we may remove unused recursive definitions ($r_1, r_2$) to get an equivalent substitution $\Gamma_4$.

$\Gamma_4  =\{
 S\hrec\,z\,w \doteq \recvar{r_3}\, w,
  \recvar{odd}=_d \mi{get}\, (\lambda x.\, \recvar{even}), \\
  \recvar{even}=_d \mi{get}\, (\lambda x.\, \recvar{r_3}\, x), 
  \recvar{r_3} =_d \lambda x.\, \mi{put}\, (\recvar{r_4}\, x)\, \recvar{odd},
  \recvar{r_4} =_d \lambda x.\, x
\}$

\end{enumerate}

Note that the final substitution $\Gamma_4$ is equivalent to the following substitution $\Gamma_c$, written in the concrete syntax 
without flattened definitions. We can easily check that this is a unifier for the concrete unification context $\Delta_c$ in Section~\ref{sec:problem_formulation}.

$\Gamma_c  =\{
 S\hrec\,z\,w \doteq \recvar{r_3}\, w,
  \recvar{odd}=_d \mi{get}\, (\lambda x.\,\mi{get}\, (\lambda y.\, \recvar{r_3}\, y)), 
  \recvar{r_3} =_d \lambda w.\, \mi{put}\, w\, \recvar{odd},
\}$

\subsection{Additional Examples}
\label{sec:ho_additional_examples}

Recall the definition for $\mt{odd}$ and $\mt{even}$ in Section~\ref{sec:problem_formulation}.
\begin{verbatim}
odd : sp = get ([x] even).
even : sp = get ([x] put x odd).
\end{verbatim}

We have seen an example of how unification can figure out
 the behavior of the stream processor $\mt{S}$
after reading two elements of the input, if it behaves the same as $\mt{odd}$.
\begin{verbatim}
?- get ([x] get ([y] S x y)) = odd.
\end{verbatim}
\[\Delta_c = \{\mi{get}\, (\lambda x.\, \mi{get}\, (\lambda y. \, S\, x\, y)) \doteq \recvar{odd},
 \recvar{odd}=_d \mi{get}\, (\lambda x.\, \recvar{even}),
 \recvar{even}=_d \mi{get}\, (\lambda x.\, \mi{put}\, x\, \recvar{odd})
 \}
 \]

 \subsubsection{An Example with no Solution}
 In the above example, $S$ may depend on both numbers $x$ and $y$ that read from the input stream. 
 We may restrict $x$ to only use the number at index $0$, by omitting $y$ from the argument of $S$.
\begin{verbatim}
?- get ([x] get ([y] S x)) = odd.
\end{verbatim}
\[\Delta_c = \{\mi{get}\, (\lambda x.\, \mi{get}\, (\lambda y. \, S\, x)) \doteq \recvar{odd},
 \recvar{odd}=_d \mi{get}\, (\lambda x.\, \recvar{even}),
 \recvar{even}=_d \mi{get}\, (\lambda x.\, \mi{put}\, x\, \recvar{odd})
 \}
 \]
 The $\m{odd}$ stream processor outputs an element at index $1$,  but $S$ doesn't have access 
 to $y$. This unification problem does not have a solution, and the algorithm eventually adds the symbol $\m{contra}$ to the unification context.
 The first six steps (equations and definitions (1) - (13)) are 
 similar to the previous example with changes marked in \blue{blue}, but now the equation $(13)$ is no longer a resolution equation for $S\hrec$.
 Note that $\Delta_c \rhd \Delta_{7}$ (equations and definitions $(1)-(7)$).

\begin{longtable}{ll}
$(1)\, \recvar{r_1} \doteq \recvar{odd} $ & given\\
$(2)\, \recvar{r_1} =_d \mi{get}\, (\lambda  x.\, \recvar{r_2}\,  x) $ \\
$(3)\, \recvar{r_2} =_d \lambda x.\, \mi{get}\, (\lambda  y.\, \blue{S\hrec\, x}) $ \\
$(4)\, \recvar{odd}=_d \mi{get}\, (\lambda x.\, \recvar{even}) $ \\
$(5)\, \recvar{even}=_d \mi{get}\, (\lambda x.\, \recvar{r_3}\, x) $ \\
$(6)\, \recvar{r_3} =_d \lambda x.\, \mi{put}\, (\recvar{r_4}\, x)\, \recvar{odd}  $ \\
$(7)\, \recvar{r_4} =_d \lambda x.\, x $ \\
$(8)\, \mi{get}\, (\lambda x.\, \recvar{r_2}\, x) \doteq \mi{get}\, (\lambda x.\, \recvar{even})$ & by Rule (\ref*{rule:ho_rec_expansion}) on $(1)$, $(2)$ and $(4)$ \\
$(9)\,\lambda x.\, \recvar{r_2}\, x \doteq \lambda x.\, \recvar{even}$ & by Rule (\ref*{rule:ho_structural}) on $(8)$ \\
$(10)\, \recvar{r_2}\, z \doteq \recvar{even}$ & by Rule (\ref*{rule:ho_inst_rec}) on $(9)$, and we verify that there does not \\ &exist any equation $(\exists x)\recvar{r_2}\, x \doteq \recvar{even}$ in the context $\Delta_9$.  \\
$(11)\, \mi{get}\, (\lambda y.\, \blue{S\hrec\, z}) \doteq \mi{get}\, (\lambda x.\, \recvar{r_3}\, x)$ & by Rule (\ref*{rule:ho_rec_expansion}) on $(10)$, $(3)$ and $(5)$ \\
$(12)\,\lambda y.\, \blue{S\hrec\, z}\doteq\lambda x.\, \recvar{r_3}\, x$ & by Rule (\ref*{rule:ho_structural}) on $(11)$ \\
$(13) \,\blue{S\hrec\,z} \doteq \recvar{r_3}\, w$ & by Rule (\ref*{rule:ho_inst_rec}) on $(12)$, and we verify that there does not \\ &exist any equation $(\exists x)\,S\hrec\,z \doteq \recvar{r_3}\, x$ in the context $\Delta_{12}$. \\
$(14)\, S\hrec \, z \doteq \recvar{t}$ & by Rule (\ref*{rule:ho_rec_pruning}) on $(13)$ \\
 $(15)\, \recvar{r_3} \, w \doteq \recvar{t}$  \\
 $ (16)\, \recvar{t}=_d G\hcontra$ \\
$(17)\, \mi{put}\, (\recvar{r_4}\, w)\, \recvar{odd} \doteq G\hcontra $ & by Rule (\ref*{rule:ho_rec_expansion}) on $(15)$, $(6)$, and $(16)$ \\
$(18)\, G\hcontra \doteq \mi{put}\, (\recvar{r_4}\, w)\, \recvar{odd}   $ & by Rule (\ref*{rule:ho_symmetry_contra}) on $(17)$ \\
$ (19) \, G\hcontra \doteq \mi{put} F\hrec H\hrec $ & by Rule (\ref*{rule:ho_immitation}) on $(17)$ \\
$ (20) \, \mi{put}\, (\recvar{r_4}\, w)\, \recvar{odd} \doteq \mi{put} F\hrec H\hrec $ & by Rule (\ref*{rule:ho_transitivity_contra}) on $(17)$ and $(19)$  \\
$ (21)\, \recvar{r_4}\, w\, \doteq F\hrec$ & by Rule (\ref*{rule:ho_structural}) on $(20)$  \\
$ (22)\,\recvar{odd} \doteq H\hrec $ \\
$ (23)\, F\hrec \doteq \recvar{r_4}\, w$ & by Rule (\ref*{rule:ho_symmetry_rec}) on $(21)$  \\
$ (24)\, F\hrec \doteq \recvar{s}$ & by Rule (\ref*{rule:ho_rec_pruning}) on $(23)$  \\
$ (25) \,\recvar{r_4}\, w \doteq \recvar{s}$ \\
$ (26)\,\recvar{s}=_d F\hcontra$ \\
$ (27)\, w \doteq F\hcontra $ & by Rule (\ref*{rule:ho_rec_expansion}) on $(25)$, $(7)$, and $(26)$  \\
$ (28)\,  F\hcontra \doteq w $ & by Rule (\ref*{rule:ho_symmetry_rec}) on $(27)$  \\
$ (29)\, \m{contra}$ & by Rule (\ref*{rule:ho_projection_fail}) on $(28)$\\
\end{longtable}

We now consider some more problems that do not involve $\mi{odd}$ or $\mi{even}$.

\subsubsection{A Stream Processor that Keeps Producing Elements}

For example, we may ask, what is a stream $H$ that outputs the given element and continues as itself.
\begin{verbatim}
?- [x] put x (H x) = [x] H x.
\end{verbatim}

We should have the following result, which says that \verb$H$ produces an argument and continues as itself.
\begin{verbatim}
H = [x] put x (H x).
\end{verbatim}

Indeed, the algorithm is able to find this solution. 
Our unification algorithm correctly finds a recursive definition for $H\hrec$, as seen below, with $\Delta_c \rhd \Delta_3$
(equations and definitions $(1) -(3)$).

\[\Delta_c = \{ \lambda x.\, \mi{put}\, x\, (H\, x) \doteq \lambda x.\, H\, x\}\]

\begin{tabular}{ll}
 $ (1)\, \lambda x.\, \recvar{r_1}\, x \doteq \lambda x.\, H\hrec \, x$ & given\\ 
 $ (2)\, \recvar{r_1} =_d \lambda x.\, \mi{put}\, (\recvar{r_2}\, x)\, (H\hrec\, x)$\\
 $ (3)\, \recvar{r_2} =_d \lambda x.\, x$ \\
  $(4)\, \recvar{r_1}\, z \doteq H\hrec z$ & by Rule (\ref*{rule:ho_inst_rec}) on $(1)$ \\
  $(5)\, \dots$  & $\dots$
  \end{tabular}

   \[\m{unif}(\Delta_4) = \{H\hrec z \doteq \recvar{r_1}\, z, \recvar{r_1} =_d \lambda x.\, \mi{put}\, (\recvar{r_2}\, x)\, (\recvar{r_1}\, x), 
    \recvar{r_2} =_d \lambda x.\, x \}\]

Note that this problem has no unifier under Miller's higher-order pattern unification algorithm \cite{Miller91jlc} due to the failure of the 
occurs check on $\mt{H}$. 

\subsubsection{A Stream Processor that Keeps Consuming Elements}
    Dually, consider a stream processor $S$ that reads an element and continues as itself, with $\Delta_c\rhd\Delta_2$.
\begin{verbatim}
?- [x] get ([y] S y) = [x] S x.
\end{verbatim}

The intuitive solution is that \verb$S$ reads an element and continues as itself.
\begin{verbatim}
S = [x] get ([y] S y).
\end{verbatim}

The algorithm is able to find this solution.
\[\Delta_c = \{ \lambda x.\, \mi{get}\, (\lambda y.\, S\, y) \doteq \lambda x.\, S\, x\}\]

    \begin{tabular}{ll}
      $(1)\, \lambda x.\, \recvar{r_1}\, x \doteq \lambda x.\, S\hrec \, x$ & given\\
      $(2)\, \recvar{r_1} =_d \lambda x.\, \mi{get}\, (\lambda y.\, S\hrec\, y)$ \\
      $(3)\, \recvar{r_1}\, z \doteq S\hrec z$ & by Rule (\ref*{rule:ho_inst_rec}) on $(1)$ \\
      $(4)\, \dots$  & $\dots$
      \end{tabular}
       \[\m{unif}(\Delta_3) = \{S\hrec\, z \doteq \recvar{r_1}\, z, \recvar{r_1} =_d \lambda x.\, \mi{get}\, (\lambda y.\, \recvar{r_1}\, y) \}\]
        Here, the definition $r_1$ never uses its argument. 
        Our unification algorithm will not actively prune the arguments to recursion constants unless triggered by a unification equation like in Rule (\ref*{rule:ho_rec_pruning}).

\subsubsection{Another Stream Processor that Keeps Consuming Elements}

In the previous example, \verb$S$ discards its arguments but passes the read element to itself. One interesting question 
to ask would be, what if it discards the read element and passes the input argument to itself. Will the two stream processors be equal?

\begin{verbatim}
?- [x] get ([y] S y) = [x] S x. [x] get ([y] S x) = [x] S x.
\end{verbatim}

We expect \verb$S$ not to be able to use any arguments and instead should discard all arguments. Informally, we expect \verb$S$ to be equivalent to \verb$H$, 
which discards all arguments and just keeps reading the input stream.

\begin{verbatim}
S = [x] H. H = get ([y] H).
\end{verbatim}

The algorithm is able to find this solution. This example highlights the use of Rules (\ref*{rule:ho_rec_pruning})(\ref*{rule:ho_unif_consistency_rec}), 
during steps $(9)(15)(16)$.

\[\Delta_c = \{
  \lambda x.\, \mi{get}\, (\lambda y.\, S\, y) \doteq \lambda x.\, S\, x,
  \lambda x.\, \mi{get}\, (\lambda y.\, S\, x) \doteq \lambda x.\, S\, x
\}\]

    \begin{tabular}{ll}
      $(1)\, \lambda x.\, \recvar{r_1}\, x \doteq \lambda x.\, S\hrec \, x$ & given\\
      $(2)\, \recvar{r_1} =_d \lambda x.\, \mi{get}\, (\lambda y.\, S\hrec\, y)$ \\
      $(3)\, \lambda x.\, \recvar{r_2}\, x \doteq \lambda x.\, S\hrec \, x$  \\
      $(4)\, \recvar{r_2} =_d \lambda x.\, \mi{get}\, (\lambda y.\, S\hrec\, x)$ \\
      $(5)\, \recvar{r_1}\, z \doteq S\hrec\, z$ & by Rule (\ref*{rule:ho_inst_rec}) on $(1)$ \\
      $(6)\, \recvar{r_2}\, w \doteq S\hrec\, w$ & by Rule (\ref*{rule:ho_inst_rec}) on $(3)$ \\
      $(7)\,  S\hrec\, z\doteq \recvar{r_1}\, z $ & by Rule (\ref*{rule:ho_symmetry_rec}) on $(5)$ \\
      $(8)\,  S\hrec\, w \doteq \recvar{r_2}\, w $ & by Rule (\ref*{rule:ho_symmetry_rec}) on $(6)$ \\
      $(9)\,  \recvar{r_1}\,z \doteq \recvar{r_2}\, z $ & by Rule (\ref*{rule:ho_unif_consistency_rec}) on $(7)$ \\
      $(10)\, \mi{get}\, (\lambda y.\, S\hrec\, y) \doteq \mi{get}\, (\lambda y.\, S\hrec\, z)$ & by Rule (\ref*{rule:ho_rec_expansion}) on $(9)$, $(2)$ and $(4)$ \\
      $(11)\, \lambda y.\, S\hrec\, y \doteq \lambda y.\, S\hrec\, z$ & by Rule (\ref*{rule:ho_structural}) on $(10)$ \\
      $(12)\, S\hrec\, y \doteq S\hrec\, z$ & by Rule (\ref*{rule:ho_inst_rec}) on $(11)$ \\
      $(13)\, S\hrec\, y \doteq H\hrec$ & by Rule (\ref*{rule:ho_flexible_flexible_same_var}) on $(12)$ \\
      $(14)\, S\hrec\, z \doteq H\hrec$ \\
      $(15)\, \recvar{r_1}\, z \doteq H\hrec$ & by Rule (\ref*{rule:ho_unif_consistency_rec}) on $(7)$ and $(14)$ \\
      $(16)\, \recvar{r_1}\, z \doteq \recvar{t}$ & by Rule (\ref*{rule:ho_rec_pruning}) on $(15)$ \\
      $(17)\, H\hrec \doteq \recvar{t}$\\
      $(18)\, \recvar{t} =_d G\hcontra$\\
      $(19)\, \mi{get}\, (\lambda y.\, S\hrec\, y) \doteq G\hcontra$ & by Rule (\ref*{rule:ho_rec_expansion}) on $(16)$, $(2)$ and $(18)$ \\
      $(20)\, G\hcontra \doteq \mi{get}\, (\lambda y.\, S\hrec\, y)$ & by Rule (\ref*{rule:ho_symmetry_contra}) on $(19)$ \\
      $(21)\, S\hrec\, y \doteq \recvar{t}$ & by Rule (\ref*{rule:ho_transitivity_rec}) on $(13)$ and $(17)$ \\
      $(22)\, \dots$  & $\dots$
      \end{tabular}
      We have 
       \[\m{unif}(\Delta_{20}) = \{S\hrec\, z \doteq \recvar{t}, \recvar{t} =_d  \mi{get}\, (\lambda y.\, \recvar{t}) \}\]

       We explain the construction of $\m{unif}(\Delta_{20})$, and only show the most relevant equations and definitions.

       \begin{tabular}{ll}
         $\Gamma_1 = \{\recvar{t} =_d G\hcontra, \dots\}$ &  Start with all the recursive definitions of $\Delta_{20}$. \\
         $\Gamma_1 = \{ \blue{S\hrec\, z \doteq \recvar{t}},\recvar{t} =_d  G\hcontra,  \dots\}$ &  Add the resolution equation for $S\hrec$. \\
         $\Gamma_1 = \{S\hrec\, z \doteq \recvar{t}, \recvar{t} =_d \blue{ \mi{get}\, (\lambda y.\, S\hrec\, z)}, \dots\}$ &  Replace resolved unification metavariable $G\hrec$ \\ & with its resolution. \\
         $\Gamma_1 = \{S\hrec\, z \doteq \recvar{t},\recvar{t} =_d  \mi{get}\, (\lambda y.\, \blue{t}),  \dots\}$ &  Replace resolved unification metavariable $S\hrec$  \\ &  with its resolution. \\
       \end{tabular}

\subsubsection{Variable Dependency}
  Finally, consider the following unification problem which illustrates how pattern variables following different unification metavariables restrict 
  how the variables can be used.
\begin{verbatim}
?- get ([x] get ([y] H x)) = get ([x] get ([y] S y)).
\end{verbatim}
After reading two input elements, continuation \verb$H$ may use the first element, and continuation \verb$S$ may use the second element, 
but \verb$H$ and \verb$S$ have to be equal. Intuitively, \verb$H$ and \verb$S$ will be equal to some stream processor \verb$F$ that does not use any arguments.

\begin{verbatim}
H = [x] F, S = [x] F 
\end{verbatim}

Indeed, unification correctly finds that neither \verb$H$ nor \verb$S$ can use their argument.
\[\Delta_c = \{\mi{get}\, (\lambda x.\, \mi{get}\, (\lambda y. \, H\, x)) \doteq \mi{get}\, (\lambda x.\, \mi{get}\, (\lambda y. \, S\, y))\}
 \]

\begin{tabular}{ll}
$(1)\, \recvar{r_1} \doteq \recvar{r_3} $ & given \\
$(2)\, \recvar{r_1} =_d \mi{get}\, (\lambda  x.\, \recvar{r_2}\,  x) $\\
$(3)\, \recvar{r_2} =_d \lambda x.\, \mi{get}\, (\lambda  y.\, H\hrec\, x) $\\
$(4)\, \recvar{r_3} =_d \mi{get}\, (\lambda  x.\, \recvar{r_4}\,  x) $\\
$(5)\, \recvar{r_4} =_d \lambda x.\, \mi{get}\, (\lambda  y.\, S\hrec\, y) $\\
$ (6)\, \mi{get}\, (\lambda x.\, \recvar{r_2}\, x) \doteq \mi{get}\, (\lambda x.\, \recvar{r_4}\, x)$ & by Rule (\ref*{rule:ho_rec_expansion}) on $(1)$, $(2)$ and $(4)$ \\
$ (7)\,\lambda x.\, \recvar{r_2}\, x \doteq \lambda x.\, \recvar{r_4}\, x$ & by Rule (\ref*{rule:ho_structural}) on $(6)$ \\
$ (8)\, \recvar{r_2}\, z \doteq \recvar{r_4}\, z$ & by Rule (\ref*{rule:ho_inst_rec}) on $(7)$ \\
$ (9)\, \mi{get}\, (\lambda y.\, H\hrec\, z) \doteq \mi{get}\, (\lambda y.\, S\hrec\, y)$ & by Rule (\ref*{rule:ho_rec_expansion}) on $(8)$, $(3)$ and $(5)$ \\
$ (10)\,\lambda y.\, H\hrec\, z\doteq\lambda y.\, S\hrec\, y$ & by Rule (\ref*{rule:ho_structural}) on $(9)$ \\
$ (11) \,H\hrec\,z \doteq S\hrec\, w$ & by Rule (\ref*{rule:ho_inst_rec}) on $(10)$ \\
$ (12)\, H\hrec \, z \doteq F\hrec$ & by Rule (\ref*{rule:ho_flexible_flexible_different_vars}) on $(11)$ \\
$(13)\, S\hrec\, w \doteq F\hrec $\\
$(14)\, \dots$ & $\dots$ \\
\end{tabular}

Now the unifier for $\Delta_{13}$ is
\[\m{unif}(\Delta_{13}) = \{H\hrec \, z \doteq F\hrec, S\hrec\, w \doteq F\hrec, F\hrec \doteq F\hrec\}\]

\subsection{Correctness of the Algorithm}
\label{sec:metatheory}


As with the first-order case, we show the following three properties of the algorithm to establish its correctness.
\begin{enumerate}[(1)]
  \item \textbf{Correspondence}. At each step of the algorithm, the most general unifier of the context before corresponds to the most general unifier of the context after (Theorem~\ref{thm:ho_correspondence}).
  Some steps in the algorithm create additional unification metavariables, and as a result, this correspondence needs to take the difference in the domains of the unifiers into account.
 The correspondence proof shares the essential ideas of \citet{Miller91jlc} and \citet{Huet75tcs}'s proofs in that terms are inspected one level at a time, and the dependencies of unification metavariables on pattern arguments are implicitly kept track of.
  \item \textbf{Termination}. Any unification context always saturates in a finite number of steps (Theorem~\ref{thm:ho_termination}).
  This proof is much more involved due to the presence of pattern variables. 
  \item \textbf{Correctness}. The unifier for a saturated unification context $\m{unif}(\Delta)$ is actually the most general unifier for $\Delta$ (Theorem~\ref{thm:ho_correctness_of_unifiers}).
  The main complexity in this case is to handle the pattern variables following a unification metavariable.
\end{enumerate}

\begin{lemma}
  \label{thm:ho_unifier_subset}
  Given unification contexts $\Delta$ and $\Delta'$,
  $\m{eqs}(\Delta) \subseteq \m{eqs}(\Delta')$, $\m{defs}(\Delta) = \m{defs} (\Delta')$, $UV(\Delta) = UV(\Delta')$,
   then any unifier of $\Delta'$ is a unifier of $\Delta$.
\end{lemma}
\begin{proof}
  The proof is essentially the same as the proof for Lemma~\ref{thm:fo_unifier_subset}, by observing that 
  definitional expansion does not depend on unification equations but only on recursive definitions.
\end{proof}

\begin{lemma}
  \label{thm:ho_unifier_superset}
  If $\Gamma$ is a unifier for $\Delta$, then $\Gamma$ is a unifier for $\Delta'$
  where $\Delta'$ has all recursive definitions of $\Delta$ and additional true equations 
  $M \doteq M'$  in the sense that $\m{exp}^{\Delta[\Gamma]}\depth k (M[\Gamma]) = \m{exp}^{\Delta[\Gamma]}\depth k (M'[\Gamma])$.
\end{lemma}
\begin{proof}
  Because definitional expansion depends only on recursive definitions,
  we have 
  $\m{exp}^{\Delta [\Gamma]} \depth k (M) = \m{exp}^{\Delta' [\Gamma]} \depth k (M)$ for all $k$ and $M$.
\end{proof}

  Let ${\Gamma'}|_{UV(\Delta)}$ be the substitution $\Gamma'$ with the domain restricted to $UV(\Delta)$.
  That is, $\Gamma'|_{UV(\Delta)}$ is obtained from $\Gamma'$ by
  removing all substitution equations of $\Gamma'$ if the unification metavariable on the left-hand side is not in $UV(\Delta)$.

  We state a similar lemma to Lemma~\ref{thm:ho_unifier_superset} but now the unifier of $\Delta'$ can have a larger domain than $UV(\Delta)$.

\begin{lemma}
  \label{thm:ho_unifier_superset_dom_restr}
  If $\Gamma$ is a unifier for $\Delta$, then $\Gamma'$ (where $\Gamma = \Gamma' |_{dom(\Gamma)}$) is a unifier for $\Delta'$
  where $\Delta'$ has all recursive definitions of $\Delta$ and additional true equations 
  $M \doteq M'$  in the sense that $\m{exp}^{\Delta[\Gamma']}\depth k (M[\Gamma']) = \m{exp}^{\Delta[\Gamma']}\depth k (M'[\Gamma'])$.
\end{lemma}
\begin{proof}
  Similar to the proof of Lemma~\ref{thm:ho_unifier_superset}.
\end{proof}

We show that the domain restriction preserves unifiers and their ordering.
\begin{lemma}
  \label{thm:ho_unifier_subset_fewer_unif_vars}
  Given unification contexts $\Delta$ and $\Delta'$,
  $\m{eqs}(\Delta) \subseteq \m{eqs}(\Delta')$, $\m{defs}(\Delta) \subseteq \m{defs} (\Delta')$, $UV(\Delta) \subseteq UV(\Delta')$,
   for any unifier $\Gamma'$ of  $\Delta'$, $\Gamma'|_{UV(\Delta)}$ is a unifier of $\Delta$.
\end{lemma}
\begin{proof}
  This proof is similar to the proof of Lemma~\ref{thm:fo_unifier_subset}.

  Let $\Gamma'$ be a unifier of $\Delta'$,  all unification equations of $\Delta'[\Gamma'] $ hold.
  Let $\Gamma = \Gamma'|_{UV(\Delta)}$.
  Take any $M \doteq M' \in \Delta$, we want to show 
  $\m{exp}^{\Delta[\Gamma]}\depth k(M[\Gamma]) = \m{exp}^{\Delta[\Gamma]}\depth k(M'[ \Gamma])$.
  We know $\m{exp}^{\Delta' [\Gamma']}\depth k (M[\Gamma'])= \m{exp}^{\Delta' [\Gamma']}\depth k(M'[\Gamma'])$, 
  and all recursion constants in $M$ occur in $\Delta$ (because unification contexts have to be well-formed).
  WLOG, it suffices to show $\m{exp}^{\Delta [\Gamma]}\depth k (M[\Gamma])= \m{exp}^{\Delta' [\Gamma']}\depth k(M[\Gamma'])$.
  By definition, for any $M$ that occurs in $\Delta$, $M[\Gamma] = M[\Gamma']$, we have $\Delta[\Gamma] = \Delta[\Gamma']$.
  Then, it suffices to show $ \m{exp}^{\Delta [\Gamma']}\depth k (M[\Gamma'])
 = \m{exp}^{\Delta' [\Gamma']}\depth k(M[\Gamma'])$.
  But since definitional expansions only depend on recursive definitions that actually occur in $\Delta$, we have 

  \begin{tabular}{ll}
 $\m{exp}^{\Delta [\Gamma']}\depth k (M[\Gamma'])$ \\
 $= \m{exp}^{\m{defs}(\Delta [\Gamma'])}\depth k(M[\Gamma'])$ & (expansion only definition only depends on definitions of $\Delta[\Gamma']$) \\
 $= \m{exp}^{\m{defs}(\Delta' [\Gamma'])}\depth k(M[\Gamma'])$ & ($M$ can only depend on recursion constants occurring in $\Delta$)\\
 $= \m{exp}^{\Delta' [\Gamma']}\depth k(M[\Gamma']) $ &(expansion only definition only depends on definitions of $\Delta'[\Gamma']$) \\
  \end{tabular}

\end{proof}

\begin{lemma}[Domain Restriction Preserves Unifier Ordering]
  \label{thm:domain_restriction_preserves_most_general_unifiers}
  
  Given a substitution $\Gamma_1 \circ \Gamma_2 = \Gamma_3$, 
  let $S \subseteq \m{dom}(\Gamma_1)$, then there exists $\Gamma_2'$ such that 
   $(\Gamma_1 | _{S}) \circ \Gamma_2' = (\Gamma_3 | _{S}) $. 
   Moreover,  $\Gamma_2' = \Gamma_2 | _{FUV(\Gamma_1|_{S})}$,
where the set of  free unification metavariables of a substitution, $FUV(\Gamma)$, is the set of unification metavariables that occur on the right-hand sides and recursive definitions of the substitution $\Gamma$.
\end{lemma}
\begin{proof}
  For any substitution equation $H^m \,\bar x \doteq M \in \Gamma_1| _{S}$, 
  the result of applying $\Gamma_2$ to $M$ is the same as the result of 
  applying $ \Gamma_2 | _{FUV(\Gamma_1|_{S})}$ to $M$.

\end{proof}

\begin{theorem}
  [Correspondence]
  \label{thm:ho_correspondence}
  If $\Delta$ transforms into $\Delta'$ by applying one of the rules to some equation in $\Delta$, 
  then the set of unifiers of $\Delta$ coincides with the set of unifiers in $\Delta'$ with domains restricted to $UV(\Delta)$.
  Moreover, domain restriction preserves most general unifiers.
  
\end{theorem}
\begin{proof}
  The proof has two parts. First, we show that domain restriction preserves unifiers. Then, we show that the domain restriction preserves most general unifiers.

  \textbf{(Part 1)}

  If $\Delta'$ contains $\m{contra}$ then there is no unifier for $\Delta'$, and we can show in each case that there is no unifier for $\Delta$ 
  by inspecting rules (\ref*{rule:ho_structural_var_const})(\ref*{rule:ho_structural_var_var_clash})(\ref*{rule:ho_structural_const_const_clash})(\ref*{rule:ho_projection_fail}), and the case where $\m{contra}$ is already present in $\Delta$.
   Otherwise, assume $\m{contra} \notin \Delta'$.

  There are two groups of rules, rules that do not add new unification metavariables (\textbf{Group 1}), and rules that add new unification metavariables (\textbf{Group 2}).

  (\textbf{Group 1})
  For the rules that do not add new unification metavariables, every unifier for $\Delta'$ is also a unifier of $\Delta$ by Lemma~\ref{thm:ho_unifier_subset}.
  And every unifier of $\Delta$ is also a unifier for $\Delta'$ by Lemma~\ref{thm:ho_unifier_superset}. 
  We show rules (\ref*{rule:ho_rec_expansion})(\ref*{rule:ho_unif_consistency_contra}) as examples for applying Lemma~\ref{thm:ho_unifier_superset}.

  Rule (\ref*{rule:ho_rec_expansion}), given $\recvar{r}\,\bar x \doteq \recvar{s}\,\bar y$, $r=_d \lambda z.\, U_1$, and $s =_d \lambda w.\, U_2$ in $\Delta$, it adds the equation 
  $[\bar x/\bar z] U_1 \doteq [\bar y/\bar w] U_2$ in $\Delta'$. By Lemma~\ref{thm:ho_unifier_superset}, it suffices to show
  $\m{exp}\depth{k}^{\Delta[\Gamma]}(([\bar x/\bar z] U_1)[\Gamma]) = \m{exp}\depth{k}^{\Delta[\Gamma]}(([\bar y/\bar w] U_2)[\Gamma])$.
  Since $\Gamma$ is a unifier for $\Delta$, we have $\m{exp}\depth{k}^{\Delta[\Gamma]}((\recvar{r}\,\bar x)[\Gamma]) = \m{exp}\depth{k}^{\Delta[\Gamma]}((\recvar{s}\,\bar y)[\Gamma])$, 
  but $\m{exp}\depth{k}^{\Delta[\Gamma]}((\recvar{r}\,\bar x)[\Gamma]) = \m{exp}\depth{k}^{\Delta[\Gamma]}(\recvar{r}\, \bar x) = \m{exp}\depth{k}^{\Delta[\Gamma]}([\bar x/\bar z] (U_1[\Gamma]))$ and
  similarly for $(\recvar{s}\,\bar y)$.

  Rule (\ref*{rule:ho_unif_consistency_contra}), given $H\hcontra\, \bar x \doteq U_1$ and $H\hcontra \, \bar y \doteq U_2$ in $\Delta$,
  it adds the equation $U_1 \doteq [\bar x/\bar y]U_2$ to $\Delta'$.
  Let $\Gamma$ be a unifier for $\Delta$, we want to show $\Gamma$ is a unifier for $\Delta'$.
  By Lemma~\ref{thm:ho_unifier_superset}, it suffices to show $\m{exp}\depth{k}^{\Delta[\Gamma]}(U_1[\Gamma]) = \m{exp}\depth{k}^{\Delta[\Gamma]}(([\bar x/\bar y]U_2)[\Gamma])$.
  Suppose $H\hcontra\,\bar w \doteq U_{H\hcontra}$ is the substitution equation in $\Gamma$, and then

  \begin{tabular}{ll}
  $\m{exp}\depth{k}^{\Delta[\Gamma]}(U_1[\Gamma]) $ & \\
  $ = \m{exp}\depth{k}^{\Delta[\Gamma]}((H\hcontra\,\bar x)[\Gamma])$ &  (since $\Gamma$ is a unifier for $\Delta$ and $H\hcontra \, \bar x \doteq U_1 \in \Delta$)\\
  $ = \m{exp}\depth{k}^{\Delta[\Gamma]}([\bar x/\bar w]U_{H\hcontra})$ & (by the substitution equation $H\hcontra\,\bar w \doteq U_{H\hcontra} \in \Gamma$) \\
  $ = \m{exp}\depth{k}^{\Delta[\Gamma]}([\bar x/\bar y][\bar y/\bar w]U_{H\hcontra})$ & (by the definition of variable renaming)\\
  $ = [\bar x/\bar y]\m{exp}\depth{k}^{\Delta[\Gamma]}([\bar y/\bar w]U_{H\hcontra})$ & (variable renaming commutes with definitional expansion )\\
  $ = [\bar x/\bar y]\m{exp}\depth{k}^{\Delta[\Gamma]}((H\hcontra \, \bar y)[\Gamma])  $ & (by the substitution equation $H\hcontra\,\bar w \doteq U_{H\hcontra} \in \Gamma$) \\
  $ = [\bar x/\bar y]\m{exp}\depth{k}^{\Delta[\Gamma]}(U_2[\Gamma])  $ & (since $\Gamma$ is a unifier for $\Delta$ and $H\hcontra \, \bar y \doteq U_2 \in \Delta$) \\
  $ = \m{exp}\depth{k}^{\Delta[\Gamma]}([\bar x/\bar y](U_2[\Gamma]))$ & (variable renaming commutes with definitional expansion )\\
  $ = \m{exp}\depth{k}^{\Delta[\Gamma]}(([\bar x/\bar y]U_2)[\Gamma])$ & (variable renaming commutes with substitution )\\
  \end{tabular}

  The rest of the cases except rules (\ref*{rule:ho_immitation})(\ref*{rule:ho_projection})(\ref*{rule:ho_rec_pruning})(\ref*{rule:ho_flexible_flexible_different_vars})(\ref*{rule:ho_flexible_flexible_same_var}) are similar.
  It is obvious that the identity mapping (as a trivial domain restriction) preserves the most general unifiers.

  (\textbf{Group 2})
  For rules that add new unification metavariables (i.e. rules
 (\ref*{rule:ho_immitation})(\ref*{rule:ho_projection})(\ref*{rule:ho_rec_pruning})(\ref*{rule:ho_flexible_flexible_different_vars})(\ref*{rule:ho_flexible_flexible_same_var})), 
 we show that if $\Gamma'$ is a unifier for $\Delta'$, and then $\Gamma = {\Gamma'}|_{UV(\Delta)}$ is the unifier for $\Delta$.

 Rules (\ref*{rule:ho_immitation})(\ref*{rule:ho_projection}), 
 we have $H\hcontra\, \bar y \doteq h \, N_1\, \dots N_n \in \Delta$, 
 and \[H\hcontra\, \bar y \doteq h\, (\lambda \bar x_1.\, G\hrec_1\, \bar y\, \bar x_1)\, \dots\, (\lambda \bar x_n.\, G\hrec_n\, \bar y\,\bar x_n) \in \Delta'\]
 Given $\Gamma'$ is a unifier for $\Delta'$, it follows from Lemma~\ref{thm:ho_unifier_subset_fewer_unif_vars} that $\Gamma$ is a unifier for $\Delta$.
 To show that the restriction $-|_{UV(\Delta)}$ preserves unifiers, it suffices to show that given any $\Gamma$ that is a unifier of $\Delta$, 
 there exists a unique $\Gamma'$ such that $\Gamma'|_{UV(\Delta)}$ is $\Gamma$. 
 Given a unifier $\Gamma$, the substitution equation for $H\hcontra$ is $H\hcontra\, \bar z \doteq h\, (\lambda \bar x_1.\, N_1')\, \dots (\lambda \bar x_n.\, N_n')$.
 The corresponding $\Gamma'$ will map $H\hcontra$ analogously, and will have 
 the substitution equation for $G\hrec_i$ as $G\hrec_i\, \bar x_i\, \bar z \doteq N_i'$.
 $\Gamma'$ is a unifier for $\Delta'$ by Lemma~\ref{thm:ho_unifier_superset_dom_restr}.
 This $\Gamma'$ is unique because any different mapping of $G\hrec_i$ (modulo definitional expansion) will make the additional equation in $\Delta'$ false.
 
 Rule (\ref*{rule:ho_rec_pruning}), we have $H\hrec \, \bar y \doteq \recvar{r} \, \bar x  \in \Delta$, 
 and \[H\hrec \, \bar y \doteq \recvar{t}\, \bar w, \recvar{r} \, \bar x \doteq \recvar{t}\, \bar w, \recvar{t} = _d \lambda \bar w.\, G\hcontra\, \bar w \in \Delta'\] with $\bar w = \bar x \cap \bar y$.
 Given $\Gamma'$ is a unifier for $\Delta'$, it follows from Lemma~\ref{thm:ho_unifier_subset_fewer_unif_vars} that $\Gamma$ is a unifier for $\Delta$.
 To show that the restriction $-|_{UV(\Delta)}$ preserves unifiers, it suffices to show that given any $\Gamma$ that is a unifier of $\Delta$, 
 there exists a unique $\Gamma'$ such that $\Gamma'|_{UV(\Delta)}$ is $\Gamma$. 
 Given a unifier $\Gamma$, the substitution equation for $H\hrec$ is $H\hrec\, \bar y \doteq \recvar{s}\, \bar w$ with $\recvar{s} =_d \lambda \bar w.\,U \in \Gamma$
 \footnote{In practice, the substitution equation may be $H\hrec\, \bar z \doteq \recvar{s}\, \bar u$, we can $\alpha$-rename it to $H\hrec \, \bar y \doteq \recvar{s}\, ([\bar y/\bar z]\bar u)$.
 It might be the case that $([\bar y/\bar z]\bar u) \subsetneq \bar w$, with $\recvar{s} =_d \lambda([\bar y/\bar z]\bar u).\, U$, but we can always construct another definition $\recvar{q} =_d \lambda \bar w.\, U$ and set $H\hrec\, \bar y \doteq \recvar{q}\, \bar w$.}. 
  $H\hrec$ cannot be mapped to another recursive unification metavariable because of the equation $H\hrec\,\bar y \doteq r\, \bar x $.
 The corresponding $\Gamma'$ will map $H\hrec$ analogously, and will have 
 the substitution equation for $G\hcontra$ as $G\hcontra\, \bar w \doteq U$.
 This $\Gamma'$ is unique because any different mapping of $G\hcontra$ (modulo definitional expansion) will make the additional equation in $\Delta'$ false.

 Rules (\ref*{rule:ho_flexible_flexible_different_vars})(\ref*{rule:ho_flexible_flexible_same_var}) follows a similar argument. We elide the full development 
 and remark that $\Gamma'$ will map the additional unification metavariable $F^m$ analogously as $\Gamma$ maps $G^m$.

 \textbf{(Part 2)}
$-|_{UV(\Delta)}$ preserves the most general unifiers since the substitution that mediates between the most general unifier $\Gamma'$ and any more specific $\Gamma_2$ is 
a substitution whose restriction mediates 
between $\Gamma'|_{UV(\Delta)}$ and $\Gamma_2|_{UV(\Delta)}$ by Lemma~\ref{thm:domain_restriction_preserves_most_general_unifiers}.

\end{proof}

\begin{lemma}[Preservation of Well-formed Unification Contexts]
  The pattern restriction, $\beta$-normal-$\eta$-long forms, and typing are respected by the unification rules. 
\end{lemma}
\begin{proof}
  Directly by analyzing the rules.
 
\end{proof}

The preservation of $\beta$-normal-$\eta$-long forms guarantees that the number of variables following any unification metavariable 
is constant throughout. We define the \emph{width} of a unification metavariable to be the number of variables following it.
For example, if $F\hcontra \, x\, y\, z$ appears in a unification context $\Delta$, then the width of $F\hcontra$ is $3$.
Similarly, we define the \emph{width} of a recursion constant to be the number of variables following it.
A recursion constant $r$ is \emph{pruned} if there exists an equation $r\, \bar x \doteq s\, \bar y$ such that $\bar y \subsetneq \bar x$, and $r$ is \emph{unpruned} otherwise.
Note that the formal concept of $\emph{pruned}$ and $\emph{unpruned}$ recursion constants here should be understood in a different sense from the rule (\ref*{rule:ho_rec_pruning}), which has been suggesting the informal meaning of removing variables from the pattern arguments.
\begin{theorem}
  [Termination]
  \label{thm:ho_termination}
   The algorithm always terminates.
\end{theorem}
\begin{proof}
  We observe that terms in the unification contexts are shallow as defined by the grammar, and all terms are well-typed.
  Given a bounded amount of variables, unification metavariables, and recursion constants, 
  there can only be finitely many equations and recursive definitions in a unification context because terms are shallow that they are only one level deep.
  The rules that create new unification metavariables, variables, or recursion constants are rules (\ref*{rule:ho_inst_contra})(\ref*{rule:ho_inst_rec})(\ref*{rule:ho_immitation})(\ref*{rule:ho_projection})(\ref*{rule:ho_rec_pruning})(\ref*{rule:ho_flexible_flexible_different_vars})(\ref*{rule:ho_flexible_flexible_same_var}), 
  it suffices to show that these rules can only be applied finitely many times.

  First we show that given a bounded amount of unification metavariables and recursion constants,
  the rules (\ref*{rule:ho_inst_contra})(\ref*{rule:ho_inst_rec}) can only be applied finitely many times.
  Since everything is well-typed, the maximum depth and width for terms are bounded. Then, there are only finitely many equations 
  modulo simultaneous variable renaming, and the subsumption $\exists \bar x$ in the conclusion of the 
  rule (\ref*{rule:ho_inst_contra})(\ref*{rule:ho_inst_rec}) prevents additional equations from being created that are 
  merely variable renaming of existing equations. 
  
  Then, it suffices to show the rules (\ref*{rule:ho_immitation})(\ref*{rule:ho_projection})(\ref*{rule:ho_rec_pruning})(\ref*{rule:ho_flexible_flexible_different_vars})(\ref*{rule:ho_flexible_flexible_same_var})
  can only be applied finitely many times.
  We associate with each unification context a lexicographic multi-set order $\langle A, B, C\rangle$, where $A, B, C$ are multisets of natural numbers defined below,
    and show that each rule that creates new unification metavariables or recursion constants strictly decreases this order.
  The multiset order \cite{Dershowitz79acm} states that a multiset of natural numbers $X$ is considered smaller than another multiset $Y$ 
  if $X$ can be obtained from $Y$ by removing a natural number $n$ and adding a finite number of natural numbers that are strictly smaller than $n$.
  The order $\langle A, B, C\rangle$ is given by
\begin{enumerate}
  \item $A = \{\m{width}(\recvar{r}) \mid   \recvar{r} =_d U\in \Delta, \recvar{r}\, \m{unpruned}\}$ is the multiset of widths of all unpruned recursion constants.
  \item  $B =  \{\m{width}(H\hcontra) \mid H\hcontra \in UV(\Delta), H\hcontra\, \m{unresolved}\}$ is the multiset of widths of all unresolved contractive unification metavariables.
  \item  $C = \{\m{width}(H\hrec) \mid H\hrec \in UV(\Delta), H\hrec\, \m{unresolved}\}$ is the multiset of widths of all unresolved recursive unification metavariables.
\end{enumerate}
For example, we could decrease the order $\langle A, B, C\rangle$ by 
resolving a contractive unification metavariable, and adding arbitrarily many recursive unification metavariables of any width.

First, it is easy to see that no rules can ever increase this order except rules (\ref*{rule:ho_immitation})(\ref*{rule:ho_projection})(\ref*{rule:ho_rec_pruning})(\ref*{rule:ho_flexible_flexible_different_vars})(\ref*{rule:ho_flexible_flexible_same_var}):
once a unification metavariable is resolved, it remains resolved, and once a recursion constant is pruned, it remains pruned. Both conditions rely on the 
existence of certain equations and rules never remove equations from the unification context.
Then we show each of the rules
(\ref*{rule:ho_immitation})(\ref*{rule:ho_projection})(\ref*{rule:ho_rec_pruning})(\ref*{rule:ho_flexible_flexible_different_vars})(\ref*{rule:ho_flexible_flexible_same_var})
strictly decreases the order $\langle A, B, C\rangle$.

  Rule (\ref*{rule:ho_immitation}) or (\ref*{rule:ho_projection}) removes one unresolved contractive unification metavariable and adds 
  a finite number of recursive unification metavariables.

  Rule (\ref*{rule:ho_rec_pruning})  prunes a recursion constant and
  adds a contractive unification metavariable.

  Each of the rules (\ref*{rule:ho_flexible_flexible_different_vars}) and (\ref*{rule:ho_flexible_flexible_same_var}) resolves 
  a recursive (resp. contractive) unification metavariable and adds a recursive (resp. contractive) unification metavariable of a smaller width.

\end{proof}

\begin{theorem}
  [Correctness of Unifiers]
  \label{thm:ho_correctness_of_unifiers}
  If $\Delta$ is a saturated contradiction-free unification context,  and $\Gamma = \m{unif}(\Delta)$ is the most general unifier for $\Delta$.
\end{theorem}
\begin{proof}
  The proof largely follows the structure of the first-order case. We have two parts, the first part is to show that $\Gamma$ is a unifier, and the second part is to show that 
  $\Gamma$ is the most general unifier.

  \textbf{(Part 1)}
  To show $\Gamma$ is a unifier, 
  we need to show that $\m{dom}(\Gamma) = UV(\Delta)$, which is true by definition, and that every equation in $\Delta[\Gamma]$  holds.
  It suffices to show the following.
  \begin{enumerate}
    \item For all $U_1 \doteq U_2$ in $\Delta$, $\m{exp}^{\Delta[\Gamma]}\depth{k}(U_1[\Gamma]) = \m{exp}^{\Delta[\Gamma]}\depth{k}(U_2[\Gamma])$.
    \item For all $N_1 \doteq N_2$ in $\Delta$, $\m{exp}^{\Delta[\Gamma]}\depth{k}(N_1[\Gamma]) = \m{exp}^{\Delta[\Gamma]}\depth{k}(N_2[\Gamma])$.
  \end{enumerate}
  We show the following two claims simultaneously by lexicographic induction on ($k$, and the structure of $U$ or $N$), where claim (2) may refer to claim (1)
  without decreasing $k$.
  Both claims are trivial when $k = 0$. Consider the case when $k > 0$, we show (1) and (2) by case 
  analysis on the structure of $U_1 \doteq U_2$ and $N_1 \doteq N_2$. 

    We state some facts about the construction of $\m{unif}(\Delta)$.
    These facts build intuitions that make the proof cases easier to understand.
    We note that if the substitution equation for $H\hcontra$ is $H\hcontra\, \bar w \doteq F\hcontra\, \bar z$ ($F$ is necessarily unresolved by step (4) of the construction for $\Gamma = \m{unif}(\Delta)$), 
    then the substitution equation will also appear in $\Delta$ due to transitivity and the consistency of resolution rule.
    Also, if the substitution equation for $H\hcontra$ is $H\hcontra\, \bar w \doteq h\, N_1\, \dots\,N_n$, 
    there exists equations (not necessarily picked as resolution equations) $H\hcontra\, \bar w \doteq h\, N_1'\, \dots\,N_n'$ (with $FV(h\, N_1'\, \dots\,N_n') \subseteq \bar w$), 
    and $N_1 \doteq N_1', \dots, N_n \doteq N_n'$ in $\Delta$, by inspecting the process of obtaining $\Gamma$ and the structural rule (\ref*{rule:ho_structural}).
    For recursive unification metavariables $H\hrec$, the substitution equations $H\hrec\, \bar w \doteq F\hrec\, \bar z$ and $H\hrec\, \bar w \doteq \recvar{r}\, \bar z$
    will appear in $\Delta$ due to transitivity and the consistency of resolution (Rule (\ref*{rule:ho_unif_consistency_rec})).

  \begin{enumerate}[(a)]
    \item If $U_1$ or $U_2$ contains top-level $\lambda$-abstractions, then they must have equal number of $\lambda$-abstractions due to typing and $\eta$-long forms.
    Due to the commutation of the definitional expansion and $\lambda$-abstractions and the 
    commutation between substitution and $\lambda$-abstractions,  the result follows by induction hypothesis.
    The cases (b)(c)(d) below show the claim when $U_1$ and $U_2$ do not have top-level $\lambda$-abstractions.
  
  \item Both $U_1$ and $U_2$ have constructors or variables as their heads. Since $\m{contra} \notin \Delta$, they must 
   have identical constructor or variable heads. Now let $U_1 = h\, N_1\, \dots \, N_n$ and $U_2 = h\, N_1'\, \dots\, N_n'$.
    Since $\Delta$ is saturated, we have $N_i \doteq N_i'$ for all $1 \le i \le n$. The result 
    then follows from the fact that each $N_i$ and $N_i'$ 
    have equal definitional expansion up to depth $k-1$ by induction hypothesis. 

    \item One of $U_1$ and $U_2$ is a unification metavariable, and the other has a constructor as its head. 
    Without loss of generality, assume $U_1 = H\hcontra\, \bar x$ and $U_2 = h\, \overline{N}$.
    We have two cases, either $FV(h\, \overline{N}) \subseteq \bar x$ or not.
    \begin{enumerate}[(i)]
      \item 
    If $FV(h\, \overline{N}) \subseteq \bar x$, then this is a resolution equation. If this is the resolution equation used 
    in $\Gamma$, then we're done. Otherwise, the resolution equation of $H\hcontra$ in $\Gamma$ may be $H\hcontra \, \bar x \doteq h \, \overline{N'}$.
    The rule (\ref*{rule:ho_unif_consistency_contra}) ensures that there is an equation between
    $U_2$ and $h\, \overline{N'}$, and the rest follows from the case (b) above.
    \item
    If $FV(h\, \overline{N}) \not\subseteq \bar x$,  
    In this case either $H\hcontra$ is resolved or unresolved.
    $H\hcontra$ cannot be unresolved, since otherwise rule (\ref*{rule:ho_immitation}) or (\ref*{rule:ho_projection}) would apply.
    If $H\hcontra$ is resolved, then $H\hcontra\, \bar x \doteq h\, \overline{N'}$ is a substitution equation, 
    and the result follows from IH on the necessary equations between $\overline{N'}$ and $\overline{N}$.

    \end{enumerate}

    \item Both $U_1$ and $U_2$ have contractive unification metavariables as their heads. 
      Let $U_1 = H\hcontra\, \bar x$ and $U_2 = G\hcontra\, \bar y$.
    Due to saturation, WLOG, there are three cases, 
    both unification metavariables are unresolved, only one is unresolved, or both are resolved. We consider them one by one.
    \begin{enumerate}[(i)]
      \item Both are unresolved. 
      If $H\hcontra$ is equal to $G\hcontra$ (they are the same unification metavariable), then $\bar x = \bar y$ (position-wise) since otherwise rule (\ref*{rule:ho_flexible_flexible_same_var}) will apply, and 
      it would have a resolution. Otherwise, suppose $H\hcontra \ne G\hcontra$,
      Since they are in the same equivalence class, for some representative unification metavariable be $F\hcontra$, we have
      the representative equations $H\hcontra\, \bar w \doteq F\hcontra\, \bar z$ and $G\hcontra\, \bar u \doteq F\hcontra\, \bar z$ in $\Gamma$.
      Here $\bar w$ and $\bar x$ may differ.
      By rule (\ref*{rule:ho_unif_consistency_contra}), on $H\hcontra$, we have equations $G\, \bar y \doteq [\bar x/\bar w] (F\hcontra\, \bar z)$ and similarly $H\hcontra\, \bar x \doteq [\bar x/\bar w] (F\hcontra\, \bar z)$ in $\Delta$
      By symmetry and transitivity, we have $[\bar x/\bar w] (F\hcontra \, \bar z) \doteq [\bar y/\bar u] (F\hcontra \, \bar z)$ in $\Delta$. We have just shown that  
      $[\bar x/\bar w] (F\hcontra \, \bar z)$ and $[\bar y/\bar u] (F\hcontra \, \bar z)$ are syntactically equal (otherwise they will be resolved by rule (\ref*{rule:ho_flexible_flexible_same_var})).
      But now 
      $U_1[\Gamma] = (H\hcontra\, \bar x)[\Gamma] = [\bar x/\bar w](F\hcontra\, \bar z) 
      = [\bar y/\bar u](F\hcontra\, \bar z) 
       = (G\hcontra\, \bar y)[\Gamma] 
      = U_2[\Gamma] 
      $. 
      \item Only one of them is unresolved. WLOG,  $H\hcontra$ is unresolved and $G\hcontra$ is resolved.
      We have $G\hcontra\,\bar u\doteq F\hcontra \, \bar z$ or $G\hcontra\,\bar u\doteq h\, \overline N$ in $\Gamma$.
      In the first case $F\hcontra\, \bar z$ is unresolved (otherwise it would have been replaced in step (4)), 
      and then $H\hcontra$ and $F\hcontra$ are in the same equivalence class, and 
       the rest follows from the case (d)(i) above. 
       In the second case, we would have an equation $H\hcontra\,\bar x \doteq [\bar y/\bar u](h\, \overline N)$.
       But now $H\hcontra$ could be resolved by rule (\ref*{rule:ho_immitation}) or (\ref*{rule:ho_projection}).

       \item Both are resolved.
       Suppose $H\hcontra\, \bar x \doteq U_{H\hcontra}$ and $G\hcontra\, \bar y \doteq U_{G\hcontra}$ are substitution equations in $\Gamma$.
       It cannot be the case that only one of $U_{H\hcontra}$ and $U_{G\hcontra}$ has unresolved unification metavariables as the head, and the other has a constructor or a variable as the head.
       Since transitivity and rule (\ref*{rule:ho_unif_consistency_contra}) ensure an equation between $U_{H\hcontra}$ and $U_{G\hcontra}$, 
       and the other would be resolved by rule (\ref*{rule:ho_immitation}) or (\ref*{rule:ho_projection}).
       Thus, both $U_{H\hcontra}$ and $U_{G\hcontra}$ have unresolved unification metavariables as the head, 
       or both have constructors or variables as the head.
       In the first case, the result follows from the case (d)(i) above. In the second case, 
       let $U_{H\hcontra} = h\, \overline{N}$ and $U_{G\hcontra} = h\, \overline{N'}$, 
       there are equations $U_1 \doteq h\, \overline{N''}$ and $U_2 \doteq h\, \overline{N'''}$ in $\Delta$, 
       with equations between $\overline{N}$ and $\overline{N''}$ (pairwise), and similarly for $\overline{N'}$ and $\overline{N'''}$.
       By transitivity, there is an equation $h\, \overline{N''} \doteq h\, \overline{N'''}$, and thus there are equations between $\overline{N''}$ and $\overline{N'''}$ (pairwise).
       Then the result follows by IH to show $\overline{N}, \overline{N'}, \overline{N''}, \overline{N'''}$ all have equal definitional expansions up to depth $k-1$.

    \end{enumerate}

    \item The case where $N_1$ or $N_2$ contains top-level $\lambda$-abstractions is similar 
    to the case (a), and we show subsequently the cases when $N_1$ and $N_2$ do not contain top-level $\lambda$-abstractions.

    \item Both $N_1$ and $N_2$ have recursion constants as heads. Let $N_1 = \recvar r\, \bar x$, 
    where $\recvar r =_d \bar w.\, U_1 \in \Delta$ and $N_2 = \recvar s\,\bar y$, where $s =_d \lambda \bar u.\, U_2 \in \Delta$.
    Since $\Delta$ is saturated, $[\bar x/\bar w]U_1 \doteq [\bar y/\bar u] U_2 \in \Delta$.
    By IH, $\m{exp}^{\Delta[\Gamma]}\depth k(([\bar x/\bar w]U_1)[\Gamma]) = \m{exp}^{\Delta[\Gamma]}\depth k (([\bar y/\bar u] U_2)[\Gamma])$, and 
    then 

    \begin{tabular}{ll}
    $\m{exp}^{\Delta[\Gamma]}\depth k (N_1[\Gamma])  $ & \\
    $= \m{exp}^{\Delta[\Gamma]}\depth k ((\recvar r\, \bar x)[\Gamma]) $ &  (given)\\
    $= \m{exp}^{\Delta[\Gamma]}\depth k (([\bar x/\bar w]U_1)[\Gamma]) $ & (by the definition of $\m{exp}^{\Delta[\Gamma]}\depth k $)\\
    $= \m{exp}^{\Delta[\Gamma]}\depth k (([\bar y/\bar u] U_2)[\Gamma]) $ & (shown by IH above) \\
    $= \m{exp}^{\Delta[\Gamma]}\depth k ((\recvar s\, \bar y)[\Gamma]) $ & (by the definition of $\m{exp}^{\Delta[\Gamma]}\depth k $) \\
    $= \m{exp}^{\Delta[\Gamma]}\depth k (N_2[\Gamma]) $ & (given)\\
    \end{tabular}

    \item One of $N_1$ and $N_2$ is a unification metavariable, and the other has a recursion constant as its head.
    WLOG, assume $N_1 = H\hrec\, \bar x$ and $N_2 = \recvar{r}\, \bar y$. We have two cases, either $FV(r\, \bar y) \subseteq \bar x$ or not.
    \begin{enumerate}[(i)]
      \item If $FV(\recvar r\, \bar y) \subseteq \bar x$, then this is a resolution equation. If this is the substitution equation used
      in $\Gamma$, then we're done. Otherwise, rule (\ref*{rule:ho_unif_consistency_rec}) ensures that there is an equation between
      $N_2$ and the resolution equation used in $\Gamma$, and the rest follows from the case (f) above.
      \item If $FV(\recvar r\, \bar y) \not\subseteq \bar x$,
      In this case either $H\hrec$ is resolved or unresolved.
      $H\hrec$ cannot be unresolved, since otherwise rule (\ref*{rule:ho_rec_pruning}) would apply.
      If $H\hrec$ is resolved, then let $H\hrec\, \bar x \doteq \recvar{s}\, \bar z$ be a substitution equation.
      By transitivity, there is an equation between $\recvar{s}\, \bar z$ and $\recvar{r}\, \bar y$.
      By the case (f), $\m{exp}\depth k(\recvar s\, \bar z) = \m{exp}\depth k(\recvar r\, \bar y)$,
    \end{enumerate}

    \item Both $N_1$ and $N_2$ have recursive unification metavariables as their heads.
    Let $N_1 = H\hrec\, \bar x$ and $N_2 = G\hrec\, \bar y$.
    Due to saturation, WLOG, there are three cases,
    both unification metavariables are unresolved, only one is unresolved, or both are resolved. We consider them one by one.

    \begin{enumerate}[(i)]
      \item Both are unresolved. This is exactly analogous to the case (d)(i).
      \item Only one of them is unresolved. This is exactly analogous to the case (d)(ii), except that in the case the resolution is a recursion constant, the 
      unresolved unification metavariables may be resolved by rule (\ref*{rule:ho_rec_pruning}).
      \item Both are resolved. Suppose $H\hrec\, \bar x \doteq N_{H\hrec}$ and $G\hrec\, \bar y \doteq N_{G\hrec}$ are substitution equations in $\Gamma$.
      By a similar reasoning as (d)(iii), both $N_{H\hrec}$ and $N_{G\hrec}$ have unresolved unification metavariables as heads, or 
      both have recursion constants as the heads. IN the first case, the equality can be established by (h)(i).
      In the latter case, there is an equation $N_{H\hrec} \doteq  N_{G\hrec}$ due to transitivity and the rest follows by the case (f).

    \end{enumerate}

  \end{enumerate}

  \textbf{(Part 2)}
  To show $\Gamma$ is the most general unifier, given any other unifier $\Gamma_2$ of $\Delta$, it suffices to 
  construct a unifier $\Gamma_1$ such that $\Gamma \circ \Gamma_1 =\Gamma_2$. But 
  the construction of $\Gamma_1$ is easy: $\Gamma_2$ must map resolved unification metavariables analogously as 
  $\Gamma$ (otherwise a contradiction will arise), and it may choose to map equivalence classes of
  unresolved unification metavariables freely. $\Gamma_1$ simply records how unresolved unification metavariables are mapped in $\Gamma_2$.

\end{proof}

\section{Related Work}
\label{sec:related_work}

The unification algorithm for the first-order terms was first developed by \citet{Robinson65jacm} as a procedure for implementing resolution. \citet{Jaffar84ngc} gave an efficient unification algorithm for first-order rational trees based on the system of equations presentation of \citet{Martelli82tpls}.
\citet{Huet75tcs} discovered a 
pre-unification algorithm for general higher-order terms.
Although general higher-order unification is undecidable and does not have most general unifiers \cite{Huet73ic}, \citet{Miller91jlc}
discovered the pattern restriction that if arguments to 
unification metavariables are restricted to pairwise distinct bound variables, decidability and most general unifiers can be recovered.
A similar idea of restricting the arguments to bound variables gives a formulation of regular B\"ohm trees \cite{Huet98mscs} with decidable equality. 
Huet's idea later becomes the prepattern restriction in CoLF \cite{Chen23fossacs}. 
The prepattern restriction is slightly relaxed over the pattern restriction that 
repetition of bound variables are allowed.
Our use of a signature for representing recursive definitions
directly follows that of CoLF \cite{Chen23fossacs}.

Nominal unification is an alternative way of carrying out higher-order unification \cite{Urban04tcs,Urban10unif}.
It is encodable in higher-order pattern unification and higher-order pattern unification can be encoded
in nominal unification.
\citet{Schmidt-Schauss22fi} have presented a nominal unification algorithm for a version 
of cyclic $\lambda$-calculi by \citet{Ariola97tacs}. However, their cyclic $\lambda$-calculi 
has a different criterion for term equality than ours.
\section{Conclusion}

We have presented a saturation-based unification algorithm for finding most general unifiers for higher-order rational terms ($\bot$-free regular B\"ohm trees).
We have shown the termination, soundness, and completeness of this algorithm.
The main complexity is  
to arrange the conditions for applying the rules to ensure termination.
We once again find Miller's pattern fragment to be fundamental in determining the most general unifiers in the presence of  higher-order terms.
A detailed analysis of the complexity of the algorithm and an efficient implementation will be future work.

\appendix


\bibliographystyle{plainnat}
\bibliography{citationsconf}


\end{document}